\documentclass[journal,english]{IEEEtran}

\usepackage[style=ieee, doi=false, isbn=false, url=false]{biblatex}
\usepackage{amsmath,amsfonts,algorithmic,algorithm}
\usepackage{array}
\usepackage{textcomp}
\usepackage{stfloats}
\usepackage{url}
\usepackage{hyperref}
\usepackage{verbatim}
\usepackage{graphicx}
\usepackage{amssymb}
\usepackage{mathtools}
\usepackage{amsthm}
\usepackage{xcolor}
\usepackage{physics}
\usepackage{stmaryrd}
\usepackage{bm}
\usepackage{epstopdf}
\usepackage{subfigure}
\usepackage{bbm}
\usepackage[normalem]{ulem}
\usepackage{enumitem}
\usepackage{babel,csquotes}

\allowdisplaybreaks

\DeclareMathOperator*{\argmin}{\arg\,min}

\DeclareMathOperator{\diag}{diag}

\newcommand{\ignore}[1]{}

\newcommand\herm{{\mathsf{H}}}
\newtheorem{lemma}{Lemma}
\newtheorem{prop}{Proposition}
\newtheorem{corollary}{Corollary}

\def\BibTeX{{\rm B\kern-.05em{\sc i\kern-.025em b}\kern-.08em
    T\kern-.1667em\lower.7ex\hbox{E}\kern-.125emX}}

\begin{document}

\title{Atomic Hybrid Sparse/Diffuse Channel Estimation and Cram{\'e}r-Rao Bounds Analysis}

\author{
    \IEEEauthorblockN{Lei Lyu\IEEEauthorrefmark{1}, Maxime Ferreira Da Costa\IEEEauthorrefmark{2}, and Urbashi Mitra\IEEEauthorrefmark{1}}\\
    \IEEEauthorblockA{\IEEEauthorrefmark{1}University of Southern California, Los Angeles, USA}\\
    \IEEEauthorblockA{\IEEEauthorrefmark{2}CentraleSupélec, Université Paris--Saclay, Gif-sur-Yvette, France
    \thanks{This work has been funded in part by one or more of the following: ANR-24-CE48-3094, ARO W911NF1910269, ARO W911NF2410094,  NSF CIF-2311653, NSF CIF-2148313, NSF RINGS 2148313,  ONR N00014-22-1-2363, NSF DBI-2412522, and is also supported in part by funds from federal agency and industry partners as specified in the Resilient \& Intelligent NextG Systems (RINGS) program).}}
}

\maketitle

\begin{abstract}
In this paper, an atomic hybrid sparse/diffuse (aHSD) channel model in the frequency domain is proposed. Based on a structural analysis of the resolvable paths and diffuse scattering statistics, the Hybrid Atomic-Least-Squares (HALS) algorithm is designed to estimate sparse/diffuse components with a combined atomic and $\ell_2$ regularization. A theoretical analysis of the Lagrange dual problem is conducted, and the conditions required for primal and dual solutions are provided, supporting an off-the-grid delay-time estimator. 
The Cramér--Rao Bound (CRB) analysis in this paper focuses on the estimation of the channel parameters, resulting in a bound on the aggregate channel. Lower and upper bounds for the CRB on parameters are derived as functions of the minimum separations between frequency parameters. Numerical results via simulations on synthetic and real data validate the efficacy of the HALS estimation strategy and show the improved predictive ability of the CRB analysis for the performance of HALS versus previously considered bounds. 
\end{abstract}


\section{Introduction}
Many applications within wireless communications rely on high-performance channel estimation. In several key scenarios, wireless communication channels can be modeled as finite-impulse-response filters with a sparse number of taps. In these cases, sparse approximation methods yield high performance 
\cite{choi2017compressed,bajwa2010compressed,Tsai,Beygi}.
However, purely sparse models cannot always capture the channel; in such cases, a mixed model of diffuse and sparse specular components \cite{michelusi2012uwb,michelusi2012uwb2}, is a better fit (see 
\emph{e.g.} \cite{uwbdata,jiang2023long,richter2005joint}). The hybrid sparse/diffuse  model was modified for underwater acoustic channels in \cite{jiang2023long}, where a switching mechanism was designed between sparse and mixed sparse/diffuse channels. Geometric information was incorporated into hybrid channels in the frequency domain in~\cite{santos2010modeling}. 
 
A limitation of the model  in~\cite{michelusi2012uwb,michelusi2012uwb2} is that it was proposed for uniformly sampled dictionaries (\emph{e.g.} time delay, Doppler, or angle). Recent work has shown that further improved channel estimation can be achieved if these parameters are modeled as unknown, but drawn from a \emph{continuous} set. To this end, channel estimation methods based on atomic norm denoising \cite{chandrasekaran2012convex,chi2020harnessing,bhaskar2013atomic,LiLOCMAN,Tsai,Beygi} have proven quite effective.  Here, we combine the advantages of the atomic norm for estimating sparse, specular components with a diffuse model to capture additional components of more realistic channels. We note that applying atomic norm methods directly to the channel model in \cite{michelusi2012uwb,michelusi2012uwb2} does not yield good results because the differences between strong specular multipath components and weaker diffuse components are not taken into account. To better harness the channel structure, we assume the diffuse parts lie in the dictionary spanned by uniformly spaced delays.  In this work, we refine the channel model of \cite{michelusi2012uwb,michelusi2012uwb2}, which we call the \emph{atomic} Hybrid Sparse/Diffuse (aHSD) model, and further design a new channel estimator (Hybrid Atomic Least-Squares), extending the atomic norm framework for this model. 

Given the heterogeneity of the sparse and diffuse components, directly computing a Cramér--Rao Bound (CRB) on the channel is challenging. However, computing a CRB for unknown channel parameters, which characterize the channel response, is straightforward~\cite{moore2007constrained}. The channel parameters are a mixture of real and complex numbers; thus, the complexified CRB~\cite{ollila2008cramer} is needed.
Herein, we develop novel upper and lower bounds on the CRB through recent advances on the condition number of confluent Vandermonde matrices~\cite{Ferreira2022Stability}.  We observe that analyzing the new Hybrid Atomic norm Least-Squares algorithm and the aHSD channel model is not a straightforward extension of prior approaches for sparse approximation and sparse channels.
An interesting observation is that the quality of the diffuse component estimates in HALS appears to have a strong impact on the accuracy of the aHSD channel estimation and influences the occurrence of outliers~\cite{henninger2022probabilistic, muthineni2024outlier} in the sparse component estimates.

The main contributions of this paper are as follows: 
\begin{enumerate}
    \item The  \emph{atomic} Hybrid Sparse/Diffuse(aHSD) model is developed in continuous time and is then converted into a frequency-domain representation using OFDM modulation~\cite{hwang2008ofdm}.
    \item A unified channel estimation algorithm, denoted as the \emph{Hybrid Atomic norm Least-Squares} (HALS) algorithm, is proposed. The algorithm, which combines the atomic-norm and $\ell_2$-norm regularization, estimates the sparse and diffuse components of the aHSD model.
    \item Exact recovery conditions of the channel parameters in the absence of noise and optimality conditions satisfied by primal and dual solutions in the presence of additive noise are derived for HALS.
    \item  The CRBs on the channel parameters for the aHSD model are derived. Lower and upper bounds to these CRBs are further developed 
     which explicitly consider the structure of the aHSD model. 
    \item The proposed HALS algorithm is validated on both synthetic and experimental data \cite{uwbdata} and compared to the derived CRBs. HALS is numerically shown to asymptotically achieve the CRBs, especially for sparse components.
\end{enumerate}   

This paper is organized as follows. The problem formulation is provided in Section~\ref{sec:formulation}.  In particular, Section ~\ref{sec:aHSD} describes our signal and channel model (aHSD), and Section~\ref{sec:ANM} reviews the atomic norm denoising strategy and defines the estimator. 
Exact recovery of the parameters with the atomic norm is provided in~\ref{sec:hals_noiseless} in the absence of noise. The new denoising algorithm for noisy measurements, the Hybrid Atomic Least-Square Algorithm (HALS), is introduced in Section~\ref{sec:proofs}. Theoretical results and performance bounds to assess the optimality of HALS are provided in Section~\ref{sec:proofs} and~\ref{sec:CRB}. Bounds on the CRB for the sparse/diffuse channel model are established in Section~\ref{sec:CRB}. The algorithm is validated numerically in Section~\ref{sec:numerical} using simulated synthetic channels as well as real channel traces \cite{uwbdata}. The performance of HALS is compared to the derived bounds in Section~\ref{sec:CRB}. Conclusions are drawn in Section~\ref{sec:conclusions}. Key proofs are provided in the Appendix.

The following notational conventions are adopted throughout the paper: lower case $x$ denotes scalars, boldface lower case $\bm{x}$ denotes vectors, a variant symbol of the boldface lower case letter $\vb{\bm{x}}$ denotes random vectors and boldface upper case $\bm{X}$ denotes matrices. The vector from $t_1$-th to $t_2$-th entry of $\bm{x}$ is denoted as $\bm{x}_{t_1:t_2}$, where $x_t \triangleq \bm{x}_{t:t}$. The matrix from $t_1$-th to $t_2$-th row and $k_1$-th to $k_2$-th column of $\bm{X}$ is denoted $\bm{X}_{t_1:t_2, k_1:k_2}$, where $X_{t, k} \triangleq \bm{X}_{t:t, k:k}$. 
We denote $x^\star$, $\bm{x}^\star$, and $\bm{X}^\star$ as the {\bf ground truth} of scalar, vector, and matrix parameters that need to be estimated. $\hat{x}\left(\cdot\right)$, $\hat{\bm{x}}\left(\cdot\right)$ and $\hat{\bm{X}}\left(\cdot\right)$ denote estimates of scalar, vector and matrix parameters, respectively. $\left(\cdot\right)^{\top}$, $\overline{\left(\cdot\right)}$, $\left(\cdot\right)^{\herm}$, and $\left(\cdot\right)^\dagger$ are operators of transpose, conjugate, conjugate transpose, and pseudoinverse. 
The integer-valued interval is denoted by $\llbracket a, b\rrbracket =\left\{c \in \mathbb{Z}:a\leq c \leq b\right\}$.

\section{Problem Formulation}
\label{sec:formulation}
\subsection{Atomic Hybrid Sparse/Diffuse Channel Model}
\label{sec:aHSD}
We {\bf adapt} the aHSD model from~\cite{michelusi2012uwb,michelusi2012uwb2}.  We assume that the channel impulse response is decomposed as,
\begin{equation}
    h(t) = h_s(t)+h_d(t),
    \label{eq:aHSD_model_simple}
\end{equation}
where $h_s(t)$ and $h_d(t)$ represent the sparse and diffuse components, respectively. Denote by $m$ and $L$ the number of resolvable multipath components and the number of channel taps, respectively, with $m \ll L$. In order to exploit the hybrid sparse-diffuse model, we impose that the {\bf diffuse} components occur densely on a uniformly tapped delay line as in~\cite{michelusi2012uwb,michelusi2012uwb2}. The overall channel is given by
\begin{align}
    h_s(t)&= \sum_{i=1}^m\alpha_i^\star \delta(t-\tau_i^\star); & h_d(t)&=\sum_{r=0}^{L-1}\gamma_{r+1}^\star \delta(t-r\Delta T),
    \label{eq:aHSD_model}
\end{align}
where $\delta\left(\cdot\right)$ is the Dirac delta function, $\alpha_i^\star, \gamma_r^\star \in \mathbb{C}$ are the complex gains of the corresponding paths; $\tau_i^\star \in \mathbb{R}^+$ is the arbitrary delay time of the $i$-th resolvable path; and $\Delta T \in \mathbb{R}^+$ is the bin interval for the channel taps.  In contrast, the channel model in \cite{michelusi2012uwb,michelusi2012uwb2} assumes that the sparse components can only occur sparsely on a uniformly spaced tapped delay-line.

To estimate the channel, it is assumed that the transmitters send $G$ pilot signals through an OFDM signaling scheme~\cite{hwang2008ofdm}. Channel estimation is done via pilot symbols; the $g$-th pilot signal is given by,
\begin{equation}
    x^{(g)}(t)=\sum\limits_{k=0}^{N-1} s_k^{(g)} e^{j2\pi \left(k-\frac{N-1}{2}\right) \Delta f\, t}, \quad 0 \leq t < T_s,
    \label{eq:ofdm_transmit}
\end{equation}
where the odd integer $N$ denotes the number of subcarriers, $s_k^{(g)} \in \mathbb{C}$ is the transmitted symbol over the $k$-th subcarrier, $\Delta f = \frac{1}{T_s} \in \mathbb{R}$ is the subchannel space and $T_s = L\Delta T \geq \max\left\{\tau_i^\star\right\}_{i=1}^m$ is the symbol duration. Assume that a cyclic prefix (CP) is employed over $\left[-T_c, 0\right]$ before the transmission, with $T_c$ exceeding the delay spread of the channel.
Then, the signal received through the channel can be written as,
\begingroup
\allowdisplaybreaks[0]
    \begin{multline}
    y^{(g)}(t) = \sum_{i=1}^m\alpha_i^\star x^{(g)}\left(t-\tau_i^\star\right) \\+\sum_{r=0}^{L-1}\gamma_{r+1}^\star x^{(g)}\left(t-r\Delta T\right) + n^{(g)}(t),
    \label{eq:ofdm_receive}
\end{multline}
\endgroup
where $n^{(g)}(t)$ is the additive noise.

After matched filtering and sampling, the discrete-time signal received over the $k$-th subcarrier can be expressed as,
\begin{align}
    y_k^{(g)} = \frac{1}{T_s}\int_0^{T_s} y^{(g)}(t)e^{-j2\pi \left(k-\frac{N-1}{2}\right) \Delta f t}dt 
    \triangleq h_k^\star s^{(g)}_k+n^{(g)}_k. 
    \label{eq:ofdm_demodulate}
\end{align}
Stacking the signals from the subcarriers into vectors, we have
\begin{equation}
    \bm{y}^{(g)} = \bm{S}^{(g)}\bm{h}^\star +\bm{n}^{(g)},
   \label{eq:ofdm_relation}
\end{equation}
where $\bm{y}^{(g)}=\left[y_0^{(g)}, y_1^{(g)}, \cdots, y_{N-1}^{(g)}\right]^{\top}$, $\bm{S}^{(g)} = \diag\left(s_0^{(g)}, s_1^{(g)}, \cdots, s_{N-1}^{(g)}\right)$, $\bm{h}^\star=\left[h_0^\star, h_1^\star, \cdots, h_{N-1}^\star\right]^{\top}$ and $\bm{n}^{(g)}=\left[n_0^{(g)}, n_1^{(g)}, \cdots, n_{N-1}^{(g)}\right]^{\top}$.

Next, we further refine the expression of the channel $\bm h^\star$ to distinguish between the contributions of the sparse components and the diffuse ones. First, define the steering vector, $   \bm{a}\left(f\right)$,
\begin{equation}
    \bm{a}\left(f\right) \triangleq \left[ e^{-j2\pi \frac{N-1}{2}f}, e^{-j2\pi \frac{N-3}{2}f}, \cdots, e^{j2\pi \frac{N-1}{2}f} \right]^{\top} \hspace{-8pt}, \; f\in (0,1]
    \label{eq:atom_nophase}
\end{equation}
The contributions from the sparse components in the aHSD channel can be expressed as
\begin{equation}
    \bm{h}_s^\star 
    = \sum\limits_{i=1}^m \alpha_i^\star \bm{a}\left(\frac{T_s-\tau_i^\star}{T_s}\right) \triangleq \sum\limits_{i=1}^m \alpha_i^\star \bm{a}\left(f_i^\star\right),
    \label{eq:sparse_channel}
\end{equation}
where $f_i^\star\triangleq \frac{T_s-\tau_i^\star}{T_s} \in \left(0, 1\right]$. Assuming $m \ll N$, the expression of $\bm{h}_s^\star$ in Equation \eqref{eq:sparse_channel} is sparse over the set $\left\{\bm{a}
\left(f\right):f \in(0,1]\right\}$. By defining $\bm{A}_{\bm{f}^\star} =\left[\bm{a}\left(f_1^\star\right), \bm{a}\left(f_2^\star\right), \cdots, \bm{a}\left(f_m^\star\right)\right]$ and $\bm{\alpha}^\star = \left[\alpha_1^\star, \alpha_2^\star, \cdots, \alpha_m^\star\right]^{\top}$, the sparse components can further be compactly represented as $\bm{h}_s^\star = \bm{A}_{\bm{f}^\star}\bm{\alpha}^\star$.

Similar to sparse components, diffuse components admit decomposition on the steering vector in  Equation \eqref{eq:atom_nophase} of the form
\begin{align}
    \bm{h}_d^\star& = \sum\limits_{r=0}^{L-1} \gamma_{r+1}^\star \bm{a}\left(-\frac{r\Delta T}{T_s}\right) =\sum\limits_{r=0}^{L-1} \gamma_{r+1}^\star \bm{a}\left(-\frac{r}{L}\right),
    \label{eq:diffuse_channel}
\end{align}
where we assume $T_s = L\Delta T$.

The diffuse components represent the aggregate effect of many weak paths spread across arbitrary delays. Although one could model all paths (specular and diffuse) in a single framework with arbitrary delays, this would blur the distinction between strong specular reflections and diffuse scattering.  This approach would yield a dense, non-sparse, channel representation which would reduce the effectiveness of purely sparse approaches, as will be seen in the numerical results presented herein (see Section~\ref{sec:numerical}).

As a tradeoff between computational tractability and model accuracy, we assume that the diffuse components appear at every channel tap, and construct a basis matrix as follows
\begin{equation}
    \bm{D} = \left[\bm{a}\left(0\right), \bm{a}\left( -\frac{1}{L}\right), \cdots,\bm{a}\left(-\frac{L-1}{L}\right)\right]\in \mathbb{C}^{N\times L},
    \label{eq:diffusebasis}
\end{equation}
with which we can represent the diffuse components by $\bm{h}_d^\star = \bm{D}\bm{\gamma}^\star$, where $\bm{\gamma}^\star \triangleq \left[\gamma_1^\star, \gamma_2^\star, \cdots, \gamma_L^\star\right]^{\top}$.
As in \cite{michelusi2012uwb}, we will assume that the diffuse components are of much lower energy than the specular multipath, thus $\left\Vert\bm{\gamma}^\star\right\Vert_2^2$ will be small. This representation captures the spreading and low-energy characteristics of the diffuse components. 

Pre-processing Equation \eqref{eq:ofdm_relation}, we have the $g$-th snapshot expressed as 
\begin{equation}
\bm{S}_{g}^{-1}\bm{y}^{(g)} = \bm{A}_{\bm{f}^\star}\bm{\alpha}^\star+\bm{D}\bm{\gamma}^\star+\bm{S}_g^{-1}\bm{n}^{(g)},
    \label{eq:channel_model}
\end{equation}
for $g \in \llbracket 1, G\rrbracket$, where $G$ is the total number of snapshots. 

To pose our needed optimization, we define
\begin{align}
    \widetilde{\bm{y}} \triangleq \frac{1}{G}\sum\limits_{g=1}^G\bm{S}_g^{-1}\bm{y}^{(g)}
    &= \bm{A}_{\bm{f}^\star}\bm{\alpha}^\star+\bm{D}\bm{\gamma}^\star+ \widetilde{\bm{n}},\label{eq:avg_channel_model}
\end{align}
where $\widetilde{\bm{n}}\triangleq \frac{1}{G} \sum_{g=1}^G\bm{S}_g^{-1}\bm{n}^{(g)}$ is the averaged, filtered noise and $\widetilde{\bm{y}}$ is the averaged, filtered received signal, since
$\argmin\limits_{\bm{h}}\sum\limits_{g=1}^G\left\Vert\bm{S}_g^{-1}\bm{y}^{(g)}-\bm{h}\right\Vert_2^2 = \argmin\limits_{\bm{h}} \left\Vert\widetilde{\bm{y}}-\bm{h}\right\Vert_2^2= \widetilde{\bm{y}}$, the unstructured, least-squares estimate of the channel.

\subsection{Atomic Norm Denoiser}
\label{sec:ANM}
We review the classical atomic norm denoiser to set the stage for our proposed method for estimating the aHSD channel in Equation \eqref{eq:channel_model}.
The atomic norm, according to \cite{chi2020harnessing}, can be defined as
\begin{align}
    \left\Vert\bm{x}\right\Vert_\mathcal{A} \triangleq \inf\left\{\sum\limits_{i}c_i:\bm{x}=c_i\bm{a}_i, c_i>0, \bm{a}_i\in \mathcal{A}\right\},
    \label{atomicrewrite}
\end{align}
where $\mathcal{A}$ is some atomic set. With the steering vector in Equation \eqref{eq:atom_nophase}, we will use herein an atomic set of the form
\begin{align}
    \mathcal{A} &= \left\{\bm{a}_{f, \phi}=e^{j2\pi \phi}\bm{a}\left(f\right): f\in\left[0, 1\right], \phi \in \left[0, 1\right]\right\},
    \label{eq:atomicset}
\end{align}
where $\bm{a}\left(f\right)$ is given in Equation  \eqref{eq:atom_nophase}.
The corresponding dual norm $\left\Vert\bm{\cdot}\right\Vert^*_\mathcal{A}$ is given by
\begin{equation}
    \left\Vert\bm{z}\right\Vert^*_\mathcal{A} = \sup \left\{ \text{Re}\left<\bm{z}, \bm{x}\right>, \; \left\Vert\bm{x}\right\Vert_{\mathcal{A}}\leq 1 \right\}.
\end{equation}

Notice that the expressions of $\bm{h}_s^\star$ and $\bm{h}_d^\star$ given in  Equations \eqref{eq:sparse_channel} and 
\eqref{eq:diffuse_channel} are both decomposable over the continuous dictionary $\mathcal{A}$ in Equation \eqref{eq:atomicset}. Since the delays are unknown, the support of $\bm{h}_s^\star$ is unknown. However, the support of $\bm{h}_d^\star$ {\bf is known} in our setting. Atomic norm minimization can be interpreted as the implicit extension of $\ell_1$-minimization over off-the-grid dictionaries and shown to promote sparse solutions over $\mathcal{A}$ in its canonical settings~\cite{candes2014towards, chi2020harnessing} and many of its extensions~\cite{bhaskar2013atomic,LiLOCMAN}, under the proviso of a good-enough separation between the dictionary elements~\cite{da2018tight, da2020stable}.

\section{Hybrid Atomic-Least-Squares Algorithm}\label{sec:hals}

\subsection{Exact Reconstruction in the Absence of Noise}
\label{sec:hals_noiseless}
In the absence of noise (\emph{i.e.} $\bm n^{(g)} = \bm{0}$), we demix the sparse and diffuse components by solving the problem
\begin{align}\label{eq:atomic_noiseless}
    \min\limits_{\bm \gamma} \; \norm{\widetilde{\bm{y}} - \bm D \bm \gamma}_{\mathcal{A}}.  \tag{P0}
\end{align}
The optimization problem \eqref{eq:atomic_noiseless} is convex and can be reformulated as the semidefinite program (SDP)~\cite{georgiou2007caratheodory,chandrasekaran2012convex}.
    \begin{align}
        \min_{t, \bm{\iota}, \bm{h}_s, \bm{\gamma}} & t+\iota_1\nonumber \tag{P0, SDP}
          \\
        \text{subject to\;}  & \begin{bmatrix} \text{Toep}\left( \bm{\iota}\right) & \bm{h}_s \\
\bm{h}_s^{\herm} & t  \end{bmatrix} \succeq 0 \; \text{ and } \; \bm{x} = \bm{h}_s+\bm{D} \bm{\gamma}.  \nonumber
\end{align}
where $\iota_1$ is the first entry of $ \bm{\iota}$ and $\text{Toep}\left( \bm{\iota}\right)$ is the Hermitian Toeplitz matrix of the vector $\bm{\iota}$.

It is natural to investigate whether the ground truth  $\bm{h}_s^\star$, $\bm{h}_d^\star$ can be recovered from the solutions to Program \eqref{eq:atomic_noiseless}.
A key metric that measures problem feasibility is the {\bf minimum separation} between the supports of the sparse and diffuse components~\cite{Ferreira2022Stability,ferreira2023conditionNumber}. This metric is denoted by $\delta = \min \{\delta_s, \delta_{s-d}\}$, where we define the minimum wrap-around distance between the sparse supports as $\delta_{\mathrm{s}} =\min_{i\neq j} \min_{k \in \mathbb{Z}}\{|f_i - f_j + k|\}$, and the minimum wrap-around distance between the sparse and diffuse supports as $\delta_{\mathrm{s-d}} = \min_{i} \min_{r \in \mathbb{Z}}\{|f_i - \tfrac{r}{L}|\}$.

\begin{prop}[Exact recovery]\label{prop:exact_recovery}
Suppose that $N \geq 10^3$ and that $\delta = \min \{\delta_s, \delta_{s-d}\} \geq 2.52/N$, then $\hat{\bm{\gamma}}=\bm \gamma^\star$ is the unique solution to Program \eqref{eq:atomic_noiseless}. Furthermore, one has $\hat{\bm{h}}_s = \bm y - \bm D \bm \gamma^\star=\bm h_s^\star$ and the decomposition in Equation \eqref{eq:sparse_channel} is the unique atomic decomposition of $\bm h_s^\star$.
\end{prop}

The proof of Proposition~\ref{prop:exact_recovery} is presented in Appendix~\ref{sec:proof_exact_recovery}, and relies on the analysis of optimality conditions for atomic norm minimization (ANM) presented in~\cite{fernandez2016super}, but adapted to our new channel model.

\subsection{HALS Algorithm Design and Analysis}\label{sec:proofs}
In the presence of noise, based on the channel model in Equation \eqref{eq:channel_model} and the analysis of the structure of sparse and diffuse components, the following optimization problem can be formulated to estimate the aHSD channel
\begin{align} \min\limits_{\bm{h}_s, \bm{\gamma}}&\frac{1}{2}\left\Vert \widetilde{\bm{y}}-\bm{h}_s-\bm{D}\bm{\gamma}\right\Vert_2^2 + \tau \left\Vert\bm{h}_s\right\Vert_{\mathcal{A}}+\frac{\lambda}{2}\left\Vert\bm{\gamma}\right\Vert_2^2 \tag{P1},
    \label{eq:als_primal}
\end{align}
where $\tau \geq 0$ and $\lambda \geq 0$ are hyperparameters that need to be determined.

By leveraging the SDP representation of the atomic
norm {and the property of the Lagrangian multiplier}, we can reformulate Program \eqref{eq:als_primal} as
\begin{align}
    \min\limits_{t, \bm{\iota}, \bm{h}_s, \bm{\gamma}} &\left\Vert\widetilde{\bm{y}}-\bm{h}_s-\bm{D}\bm{\gamma}\right\Vert_2^2 + \tau \left(t+\iota_1\right)
    \nonumber
    \\
    \text{subject to~}  & \begin{bmatrix} \text{Toep}\left( \bm{\iota}\right) & \bm{h}_s \\
\bm{h}_s^{\herm} & t  \end{bmatrix} \succeq 0 \; \text{ and } \; \left\Vert\bm{\gamma}\right\Vert_2^2 \leq E_d.\nonumber
\tag{P2}
\label{eq:als_primal_sdp}
\end{align}
 {The hyperparameter $E_d$ can be interpreted as the energy estimate on the diffuse components' amplitude, which should be predetermined so that the inequality $\left\Vert\bm{\gamma}\right\Vert_2^2 \leq E_d$ holds for the ground truth parameter vector $\bm{\gamma}=\bm{\gamma}^\star$}. We can solve Program $\eqref{eq:als_primal_sdp}$ using off-the-shelf convex solvers, \textit{e.g.} CVX \cite{cvx}, \cite{gb08}. We denote the optimal solution as $\hat{\bm{h}}_s$ and $\hat{\bm{\gamma}}$. Then, the frequency support estimate of the sparse components, given $\hat{\bm{h}}_s$ and $\hat{\bm{\gamma}}$, is
\begin{equation}
    \hat{\mathcal{T}} = \left\{\hat{f} \in \left(0, 1\right]:\left\vert\left<\widetilde{\bm{y}}-\hat{\bm{h}}_s-\bm{D}\hat{\bm{\gamma}}, \bm{a}\left(\hat{f}\right)\right>\right\vert = \tau\right\}.
    \label{eq:supportest}
\end{equation}
A debiased estimate $\hat{\bm{h}}_{s, db}$ of the resolvable paths~\cite{bhaskar2013atomic} can be obtained by first constructing the basis matrix $\bm{A}_{\hat{\bm{f}}}$ as $\left[\bm{a}\left(\hat{f}\right), \forall \hat{f} \in \hat{\mathcal{T}}\right] \in \mathbb{C}^{N \times \left\vert\hat{\mathcal{T}}\right\vert}$, then estimating
\begin{equation}
    \hat{\bm{h}}_{s, db} = \bm{A}_{\hat{\bm{f}}}\bm{A}_{\hat{\bm{f}}}^\dagger\left(\widetilde{\bm{y}}-\bm{D}\hat{\bm{\gamma}}\right).
\end{equation}
Note that the debiasing operation is equivalent to projecting the estimate onto the the column space of $\bm{A}_{\hat{\bm{f}}}$, where we define the projection matrix $\bm{P}_{\bm{A}_{\hat{\bm{f}}}}\triangleq\bm{A}_{\hat{\bm{f}}}\bm{A}_{\hat{\bm{f}}}^\dagger$, so that $\hat{\bm{h}}_{s,db}=\bm{P}_{\bm{A}_{\hat{\bm{f}}}}\left(\widetilde{\bm{y}}-\bm{D}\hat{\bm{\gamma}}\right)$.
The solutions to Programs \eqref{eq:als_primal} and \eqref{eq:als_primal_sdp} are our estimators which generalize the classical ANM methods for purely sparse channels (\emph{i.e.} $\bm{h}_d = \bm 0$), which can be recovered by selecting $E_d=0$.

Our key theoretical results assess the correctness of the optimization problem formulation in the context of our new hybrid channel model. 
The properties that the solutions to Program \eqref{eq:als_primal} should satisfy are given as follows.
\begin{prop}
    The optimal solutions, $\hat{\bm{h}}_s$ and $\hat{\bm{\gamma}}$, to Program \eqref{eq:als_primal} satisfy:
    \begin{subequations}
        \begin{align}
       \left\Vert\widetilde{\bm{y}}-\hat{\bm{h}}_s-\bm{D}\hat{\bm{\gamma}}\right\Vert_{\mathcal{A}}^* &\leq \tau
       \label{eq:lemma2_1}
       \\
        \bm{D}^{\herm}\left(\widetilde{\bm{y}}-\hat{\bm{h}}_s-\bm{D}\hat{\bm{\gamma}} \right)&= \lambda \hat{\bm{\gamma}}
        \label{eq:lemma2_2}
        \\
        \text{Re}\left<\widetilde{\bm{y}}-\hat{\bm{h}}_s-\bm{D}\hat{\bm{\gamma}}, \hat{\bm{h}}_s\right> &= \tau\left\Vert\hat{\bm{h}}_s\right\Vert_{\mathcal{A}}.
        \label{eq:lemma2_3}
    \end{align}
    \end{subequations}
    \label{prop:primal_property}
\end{prop}
\vspace*{-0.3in}
\begin{proof}
See Appendix~\ref{appendix:prop:primal_property}.
\end{proof}
The correctness of the support estimator in Equation \eqref{eq:supportest} can be shown using the conclusions of Proposition~\ref{prop:primal_property}. The sparse components estimate w.r.t. the estimated frequencies $\hat{f}\in\hat{\mathcal{T}}$ is 
 \begin{equation}
     \hat{\bm{h}}_s = \sum\limits_{\hat{f} \in \hat{\mathcal{T}}}c_{\hat{f}}\bm{a}\left(\hat{f}\right).
     \label{eq:sparse_est_decomp}
 \end{equation}
 
Plugging~Equation \eqref{eq:sparse_est_decomp} into~Equation \eqref{eq:lemma2_3} and massaging, we have that
\begin{eqnarray}
    \sum\limits_{\hat{f} \in \hat{\mathcal{T}}}\text{Re}\left(c_{\hat{f}}^*\left<\widetilde{\bm{y}}-\hat{\bm{h}}_s-\bm{D}\hat{\bm{\gamma}}, \bm{a}\left(\hat{f}\right)\right>\right)=\tau\sum\limits_{\hat{f} \in \hat{\mathcal{T}}}\left\vert c_{\hat{f}}\right\vert,
    \label{eq:inner_prod_exp}
\end{eqnarray}
which, together with~Equation \eqref{eq:lemma2_1}, show that to make Equality \eqref{eq:inner_prod_exp} hold, the frequency estimates should satisfy Equation \eqref{eq:supportest}.

The channel estimate from Program \eqref{eq:als_primal}, which is biased, has the squared error $ \mbox{SE}_{b}=\left\Vert\bm{h}^\star_s+\bm{D}\bm{\gamma}^\star-\hat{\bm{h}}_s-\bm{D}\hat{\bm{\gamma}}\right\Vert_2^2 \triangleq \left\Vert\bm{e}_{b}\right\Vert_2^2$. We can express the debiased error as
\begin{align}
\mbox{SE}_{db}
& =\left\Vert\bm{h}_s^\star+\bm{D}\bm{\gamma}^\star-\bm{P}_{\bm{A}_{\hat{\bm{f}}}}\left(\widetilde{\bm{y}}-\bm{D}\hat{\bm{\gamma}}\right)-\bm{D}\hat{\bm{\gamma}}\right\Vert_2^2\nonumber
    \\
    & = \left\Vert\bm{P}^\perp_{\bm{A}_{\hat{\bm{f}}}}\bm{h}_s^\star+\bm{P}^\perp_{\bm{A}_{\hat{\bm{f}}}}\bm{D}\left(\bm{\gamma}^\star-\hat{\bm{\gamma}}\right)\right\Vert_2^2+\left\Vert\bm{P}_{\bm{A}_{\hat{\bm{f}}}}\widetilde{\bm{n}}\right\Vert_2^2\nonumber
    \\
     & = \mbox{SE}_{b}-\left\Vert\bm{P}_{\bm{A}_{\hat{\bm{f}}}}\bm{e}_{b}\right\Vert_2^2+\left\Vert\bm{P}_{\bm{A}_{\hat{\bm{f}}}}\widetilde{\bm{n}}\right\Vert_2^2,
    \label{eq:decomp_mse_debiased1}
\end{align}
where $\bm{P}^\perp_{\bm{A}_{\hat{\bm{f}}}}\triangleq\bm{I}_N-\bm{P}_{\bm{A}_{\hat{\bm{f}}}}$ denotes the projection matrix onto the orthogonal complement of $\bm{A}_{\hat{\bm{f}}}$.

Ideally, the true sparse components lie wholly in the subspace of the estimated sparse components.  Interestingly, this expression suggests that the error in the diffuse components should also lie in the subspace spanned by the estimated sparse components.  As the dimension of the estimated subspace gets larger, the second error term $\left\Vert\bm{P}_{\bm{A}_{\hat{\bm{f}}}}\widetilde{\bm{n}}\right\Vert_2^2$ converges to $\left\Vert\widetilde{\bm{n}}\right\Vert^2_2$, which is the error of the simple Least-Squares (LS) solution. Therefore, we want to operate in a regime where ${\hat{\bm{h}}_s}$ is more sparse.  
\subsection{Dual Problem Formulation}
{We next analyze the dual problem and show strong duality.}
\begin{prop}[Dual Problem]The dual problem of Program \eqref{eq:als_primal} is
\begin{align}
    \max\limits_{\bm{z} \in \mathbb{C}^N} & \frac{1}{2}\left\Vert\widetilde{\bm{y}}\right\Vert_2^2 -\frac{1}{2}\left\Vert\widetilde{\bm{y}}-\bm{z}\right\Vert_2^2-\frac{1}{2\lambda}\left\Vert\bm{D}^{\herm}\bm{z}\right\Vert_2^2\nonumber
    \\
    \emph{\textrm{subject to}} \; &  \left\Vert\bm{z}\right\Vert^*_\mathcal{A} \leq \tau. \tag{P3}
    \label{eq:als_dual}
\end{align}
\label{prop:dual_prolem}
\end{prop}
\vspace*{-0.3in}
\begin{proof} 
See Appendix~\ref{appendix:prop:dual_prolem}.
\end{proof}

Strong duality holds between Programs \eqref{eq:als_primal} and \eqref{eq:als_dual}, since Slater’s condition \cite{boyd2004convex} is satisfied in the constrained formulation of Program \eqref{eq:als_primal}. {Notice that the objective function in Program \eqref{eq:als_dual} is strongly concave and the feasible subspace is convex, which indicates the uniqueness of the solution to Program \eqref{eq:als_dual}.}
\begin{corollary}
    $\hat{\bm{z}} = \widetilde{\bm{y}}-\hat{\bm{h}}_s-\bm{D}\hat{\bm{\gamma}}$ is the optimal solution to Program \eqref{eq:als_dual}.
    \label{cor:dual_primal_residue}
\end{corollary}
\begin{proof}
    See Appendix~\ref{appendix:cor:dual_primal_residue}.
\end{proof}

According to Equation \eqref{eq:lemma2_2}, the diffuse component estimator is equivalent to a two-stage procedure, where $\hat{\bm{h}}_s$ is first estimated and $\hat{\bm{\gamma}}$ is the optimal solution to the following ridge regression problem:
\begin{equation}
    \min\limits_{\bm{\gamma} \in \mathbb{C}^L} \left\Vert\widetilde{\bm{y}}-\hat{\bm{h}}_s-\bm{D}\bm{\gamma}\right\Vert_2^2+\lambda\left\Vert\bm{\gamma}\right\Vert_2^2.
\end{equation}

From Corollary~\ref{cor:dual_primal_residue}, we know that the sparse component estimate from optimizing Program \eqref{eq:als_primal} is biased due to the existence of the dual variable optimizing Program \eqref{eq:als_dual}. As expected, this noisy estimate affects the estimation of the diffuse component.  This dependence is clearly seen in the ridge regression formulation above.

Given Corollary~\ref{cor:dual_primal_residue}, the estimation error from optimizing Program \eqref{eq:als_primal} can be equivalently expressed as $\mbox{SE}_b = \left\Vert\hat{\bm{z}}-\widetilde{\bm{n}}\right\Vert_2^2$. According to the dual problem derived in Proposition~\ref{prop:dual_prolem}, it follows that the magnitude of the dual solution satisfies $\left\Vert\hat{\bm{z}}\right\Vert \propto \tau$. Therefore, the solution is more biased, as $\tau$ increases, which also affects the accuracy of $\hat{\bm h}_d$. However, according to our experiments, the impact of debiasing (performance improvement in sparse component estimation) also increases as $\tau$ increases. Thus, there is a trade-off between bias and overall performance.

\vspace*{-0.15in}
\subsection{Hyper-parameters selection: $\tau$,  $\lambda$, and $E_d$}
To obtain the best denoising rates for the estimator in Program \eqref{eq:als_primal}, the hyperparameters $\tau$ and $\lambda$ must be chosen appropriately. To that end, one must have the dual solution $\hat{\bm z}$ to Program~\eqref{eq:als_dual} closely match the noise $\widetilde{\bm{n}}$.
From Corollary~\ref{cor:dual_primal_residue} and Program~\eqref{eq:als_dual}, we see that we need to choose $\tau$ such that $\left\Vert \widetilde{\bm{n}}\right\Vert^*_{\mathcal{A}} \leq \tau$ holds with high probability.
It is known that setting $\tau \propto \sigma\sqrt{N\log{N}}$ achieves near-minmax denoising rate for the classical atomic norm denoiser~\cite{chi2020harnessing,bhaskar2013atomic}. We retain this selection criterion for the HALS estimator because the diffuse component of the channel has lower energy than the sparse one. {Ignoring the sparse components, the regularizer $\lambda$ can be approximated, based on the linear minimum mean square error estimator under white Gaussian noise, as $\lambda \simeq \frac{N \sigma^2}{\mathbb{E}\left[\left\Vert \bm \gamma^\star \right\Vert_2^2\right]}$ \cite{kay1993fundamentals}.}

Using Equations \eqref{eq:lemma2_1} and \eqref{eq:lemma2_2} from Proposition~\ref{prop:primal_property}, together with Corollary~\ref{cor:dual_primal_residue}, we can derive an upper bound for the estimate of the energy of the diffuse components as
    \begin{align}
        \left\Vert\hat{\bm{\gamma}}\right\Vert_2^2&=\left<\lambda\hat{\bm{\gamma}}, \frac{1}{\lambda}\hat{\bm{\gamma}}\right>= \left<\bm{D}^{\herm}\hat{\bm{z}}, \frac{1}{\lambda}\hat{\bm{\gamma}}\right>=\frac{1}{\lambda}\left<\hat{\bm{z}}, \bm{D}\hat{\bm{\gamma}}\right>\nonumber
        \\
        &\leq \frac{\tau}{\lambda}\left\Vert\bm{D}\hat{\bm{\gamma}}\right\Vert_\mathcal{A} \leq \frac{\tau}{\lambda}\left\Vert\hat{\bm{\gamma}}\right\Vert_1 \leq \frac{\sqrt{L}\tau}{\lambda}\left\Vert\hat{\bm{\gamma}}\right\Vert_2, \nonumber
    \end{align}
which indicates $\left\Vert\hat{\bm{\gamma}}\right\Vert_2 \leq \frac{\sqrt{L}\tau}{\lambda}$.

We remark that while we adjust $\tau$ to control the sparsity of the sparse estimate $\hat{\bm{h}}_s$, the magnitude of $\hat{\bm{\gamma}}$ is also influenced, and we need to adjust $\lambda$ accordingly. Therefore, while the previous analysis provides an appropriate estimate of the scale of the two hyperparameters, the previous selection rules for $\tau$ and $\lambda$ are, in general, not min-max optimal and may require fine-tuning to achieve a better denoising rate.

To limit the need for fine-tuning, we take advantage of Program~\eqref{eq:als_primal_sdp}, which does not need to select $\lambda$. Instead, it requires an upper bound $E_d$ on the ground truth diffuse coefficients energy $\left\Vert\bm{\gamma^\star}\right\Vert_2^2$. In practice, $\bm{\gamma}^\star$ can be assumed to be drawn from a probability distribution. Hence, $E_d$ can be selected so that the event $\left\Vert\bm{\gamma^\star}\right\Vert_2^2 \leq E_d$ holds with high probability. Fixing $E_d$, we can initialize $\tau \simeq\sigma\sqrt{N\log{N}}$ and adjust its values so that the sparsity of $\hat{\bm{h}}_s$ approaches the number of resolvable paths.

\section{Analysis of the Cramér-Rao Bounds}\label{sec:CRB}

We wish to understand the fundamental limits of channel estimation given the channel description in Section~\ref{sec:formulation}. In this section, we present a tight characterization of the Cramér-Rao lower bound (CRB) for estimating both channel parameters and channel responses.

\subsection{Cramér-Rao Bounds on the Channel Parameters}
We start in this section by deriving the Fisher information matrix (FIM) for the ground truth aHSD channel parameters $\{\bm{f}^\star, \bm{\alpha}^\star, \bm{\gamma}^\star\}$ given the observation model of Equation~\eqref{eq:channel_model}. Since the amplitudes $\bm{\alpha}^\star$ and $\bm{\gamma}^\star$ of the sparse and diffuse components, respectively, are complex, one must extend the parameter space to a complexified coordinate system $\{\bm{f}^\star, \bm{\alpha}^\star, \bm{\gamma}^\star, \overline{\bm{\alpha}^\star}, \overline{\bm{\beta}^\star} \}$ in order to derive the FIM on the parameters~\cite{ollila2008cramer,smith2005StatisticalResolution}. {For notational  clarity, we remove the $^\star$ herein and the channel has parameters $\bm{f}=\bm{f}^\star$, $\bm{\alpha}=\bm{\alpha}^\star$, and $\bm{\gamma}=\bm{\gamma}^\star$.}

\begin{lemma}\label{lemma:CRB_para}
    Assume $G$ independent observations according to Equation \eqref{eq:channel_model} and $\vb{\bm{n}}^{(g)} \sim \mathcal{CN}(\bm{0}, \sigma^2\bm{I}_N)$. Denote the parameter vector as $\bm{\theta} =\left[\bm{f}^\top, \bm{\alpha}^\top, \bm{\gamma}^\top, \bm{\alpha}^\herm, \bm{\gamma}^\herm\right]^\top\in \mathbb{R}^m \times \mathbb{C}^{2m+2L}$. Then, the FIM $\bm{J}^{\bm{\theta}}$ for $\bm{\theta}$ is given by
    \begin{align}\label{eq:fim_theta}
        \bm{J}^{\bm{\theta}} = & \sigma^{-2} \bm{U}^\herm \begin{bmatrix}
            \sum\limits_{g=1}^G\bm{S}_g^\herm \bm{S}_g &\bm{0}_{N \times N}\\
            \bm{0}_{N \times N} & \sum\limits_{g=1}^G\overline{\left(\bm{S}_g^{\herm} \bm{S}_g\right)}
        \end{bmatrix} \bm{U},\\
 \mbox{where}     \;\;\;  
            \bm{U} &= \begin{bmatrix}
                \bm{A}'_{\bm{f}} & \left[\bm{A}_{\bm{f}}, \bm{D}\right]&\bm{0}_{N \times (m+L)} \\
                \overline{\bm{A}'_{\bm{f}}} & \bm{0}_{N \times (m+L)} & \overline{\left[\bm{A}_{\bm{f}}, \bm{D}\right] }
            \end{bmatrix} \in \mathbb{C}^{2N \times (3m +2L)},\nonumber\\
            \bm{A}'_{\bm{f}}&=\bm{\Lambda}\bm{A}_{\bm{f}}\diag(\bm{\alpha})\in\mathbb{C}^{N \times m},\nonumber\\
            \bm{\Lambda} &= j\pi\diag\left(-(N-1), -(N-3), \cdots,N-1\right).
        \end{align}
\end{lemma}
\begin{proof}
    Let $\bm{h}^\diamond \triangleq \left[\bm{h}^\top, \bm{h}^\herm\right]^\top$. Given $\vb{\bm{n}}^{(g)} \sim \mathcal{CN}(\bm{0}, \sigma^2\bm{I}_N)$, we have $\bm{{\vb{y}}}^{(g)} \sim\mathcal{CN}\left(\bm{S}_g\bm{h}, \sigma^2 \bm{I}_N\right)$.
    Under the independence assumption, the FIM on $\bm{h}^\diamond$ given 
    $\{\bm{y}^{(1)}, \cdots, \bm{y}^{(G)}\}$ is 
    \begin{eqnarray}
      \bm{J}^{\bm{h}^\diamond} = \sigma^{-2}\begin{bmatrix}
            \sum\limits_{g=1}^G\bm{S}_g^\herm \bm{S}_g &\bm{0}_{N \times N}\\
            \bm{0}_{N \times N} & \sum\limits_{g=1}^G\overline{\left(\bm{S}_g^{\herm} \bm{S}_g\right)}
        \end{bmatrix}.
        \label{eq:h_fim}
    \end{eqnarray}
    From the chain rule, one has $\bm{J}^{\bm{\theta}} =\left(\frac{\partial \bm{h}^\diamond}{\partial \bm{\theta}}\right)^\herm\bm{J}^{\bm{h}^\diamond}\left(\frac{\partial \bm{h}^\diamond}{\partial \bm{\theta}}\right) $; direct calculation yields,
    \begin{equation}\label{eq:dh_dtheta}
        \frac{\partial \bm{h}^\diamond}{\partial \bm{\theta}} = \begin{bmatrix}
                \bm{A}'_{\bm{f}} & \left[\bm{A}_{\bm{f}}, \bm{D}\right]&\bm{0}_{N \times (m+L)} \\
                \overline{\bm{A}'_{\bm{f}}} & \bm{0}_{N \times (m+L)} & \overline{\left[\bm{A}_{\bm{f}}, \bm{D}\right] }
            \end{bmatrix} \triangleq \bm{U}.
    \end{equation}
   to conclude on the desired statement.
\end{proof}

The Cramér-Rao lower bound states that the covariance $\bm{K}_{\bm{\theta}}$ of any unbiased estimator satisfies $\bm{K}_{\bm{\theta}} \succeq {(\bm{J}^{\bm{\theta}})}^{-1}$ for the semi-definite order. Therefore, the crux of the problem is to control the inverse of Equation \eqref{eq:fim_theta}.
Herein, we harness the Vandermonde structure of matrices $\bm{\Lambda} \bm{A}_{\bm{f}}$, $\bm{A}_{\bm{f}}$ and $\bm{D}$ in the expression of the FIM in Equation \eqref{eq:fim_theta}, analogously to that of super-resolution settings~\cite{Ferreira2022Stability}, to propose novel statistical bounds on the CRB for the aHSD channel parameters.

The problem stability is measured, as in Proposition~\ref{prop:exact_recovery}, as a function of the minimum separation between the supports of the sparse and diffuse components. 
For clarity, the pilot matrix is assumed unitary, i.e. $\bm{S}^\herm_g \bm{S}_g = \bm{I}_N$, $\forall g\in\llbracket1, G\rrbracket$. In practice, this can be achieved by transmitting phase-shift keying (PSK) data symbols with unit energy.

\begin{prop}\label{prop:CRB_bounds}
    Assume $G$ independent observations according to Equation \eqref{eq:channel_model}, $\vb{\bm{n}}^{(g)} \sim \mathcal{CN}(\bm{0}, \sigma^2\bm{I}_N)$ and $\bm{S}_g^\herm \bm{S}_g = \bm{I}_N$, for all $g \in \llbracket 1, G\rrbracket$. Let  $\delta = \min(\delta_{\mathrm{s}}, \delta_{\mathrm{s-d}})$. Furthermore, assume $\alpha_i \neq 0$ for all $i \in \llbracket 1, m\rrbracket$. 
    If $3m + 3L \leq 2N$ and $N \delta > 2$, then
    \begin{subequations}\label{eq:CRB_bounds_theo}
    \begin{align}
        K_{\min} {|\alpha_i|^{-2}} \sigma^2 &\leq \frac{2\pi^2N (N-1)(N+1)}{{3}} \mathrm{CRB}(f_i) \nonumber \\
        & \qquad \leq  K_{\max} {|\alpha_i|^{-2}} \sigma^2,\label{eq:CRB_f}\\
        K_{\min} \frac{\sigma^2}{N} &\leq \mathrm{CRB}(
        \alpha_i) \leq K_{\max} \frac{\sigma^2}{N},\label{eq:CRB_alpha} \\
       K_{\min} \frac{\sigma^2}{N} &\leq \mathrm{CRB}(
        \gamma_r) \leq  K_{\max} \frac{\sigma^2}{N},\label{eq:CRB_gamma}
    \end{align}
    \end{subequations}
for all $i \in \llbracket 1, m \rrbracket$ and all $r \in \llbracket 1, L \rrbracket$, where $K_{\min}$ and $K_{\max}$ are the constants depending on the separation parameter  $N\delta$
\begin{align*}
    K_{\min} &= \frac{1}{G\left(1 + 2{(N \delta)}^{-1} \right)},&
     K_{\max} &= \frac{1}{G\left(1 - 2{(N \delta)}^{-1} \right)}.
\end{align*}
\end{prop}

\begin{proof}
See Appendix~\ref{appendix:prop:CRB_bounds}.
\end{proof}

Proposition~\ref{prop:CRB_bounds} provides upper and lower bounds on CRB for each parameter of the aHSD channel. The proposed bounds indicate that the CRB on $f_i$ in Equation \eqref{eq:CRB_f} decays as $N^{-1}$ when the separation parameter $N\delta = \mathcal{O}(1)$ is fixed, which corroborates with the empirical Rayleigh limit, predicting a maximum resolution power of the order $\delta \simeq N^{-1}$~\cite{lindberg2012mathematical}. Similarly, the CRBs in Equations \eqref{eq:CRB_alpha} and~\eqref{eq:CRB_gamma} decay as $N^{-1}$ under a fixed noise level $\sigma$. It is interesting to note that the bounds in Equation \eqref{eq:CRB_bounds_theo} are independent of the model orders $m$ and $L$. Finally, our lower and upper bounds of the CRBs become tighter for a large separation parameter $N\delta$.
Lemma~\ref{lemma:CRB_para} and Proposition~\ref{prop:CRB_bounds} yield a significant improvement compared to~\cite{lyu2025hybrid, lyu2025cramer}, {where approximate CRBs were derived separately for the sparse and diffuse components by modeling the other component as colored Gaussian noise.}

\subsection{Cramér-Rao Bounds on the aHSD Channel}
{In the absence of models on the channel structure, a CRB for $\bm{h}^\diamond$ can be computed as the inverse of the FIM in Equation~\eqref{eq:h_fim}. However, the structure of the aHSD channel can be exploited to provide a more accurate characterization of the achievable estimation performance of $\bm{h}^{\diamond}$, as proposed below.}
\begin{prop}\label{prop:CRB_channel}
    Assume $G$ independent observations according to Equation \eqref{eq:channel_model}, $\vb{\bm{n}}^{(g)} \sim \mathcal{CN}(\bm{0}, \sigma^2\bm{I}_N)$ and $\bm{S}_g^\herm \bm{S}_g = \bm{I}_N$ for all $g \in \llbracket 1, G \rrbracket$. If $3m+2L \leq 2N$ then any unbiased estimator $\hat{\bm{h}}^\diamond$ of $\bm{h}^\diamond$ satisfies
    \begin{equation}
        \mathbb{E}\left[\left(\bm{h}^\diamond-\hat{\bm{h}}^\diamond\right)\left(\bm{h}^\diamond-\hat{\bm{h}}^\diamond\right)^{\herm}\right] \succeq \mathrm{CRB}(\bm{h}^\diamond) = \frac{\sigma^2}{G} \bm{P}_{\bm{U}},
        \label{eq:crb_channel}
    \end{equation}
    where $\bm{P}_{\bm{U}}=\bm{U}\bm{U}^\dagger$.
\end{prop}
\begin{proof}
    The result follows immediately from Proposition~\ref{prop:CRB_bounds} and the CRB for the forward transformation of parameters~ \cite{moore2007constrained, kay1993fundamentals}, which says
    \(\mathrm{CRB}(\bm{h}^\diamond) =\left(\frac{\partial \bm{h}^\diamond}{\partial \bm{\theta}}\right)\left(\bm{J}^{\bm{\theta}}\right)^{-1}\left(\frac{\partial \bm{h}^\diamond}{\partial \bm{\theta}}\right)^\herm.
    \)
\end{proof}
By Proposition~\ref{prop:CRB_channel},
the estimation error of $\bm{h}$ is lower-bounded as $\mathbb{E}\left[\left\Vert\bm{h}-\hat{\bm{h}}\right\Vert_2^2\right] \geq \frac{1}{2}\mathop{trace}(\mathrm{CRB}(\bm{h}^\diamond)) = \frac{\sigma^2}{G}\left(\frac{3}{2}m+L\right)$, which is proportional to the number of unknown parameters.
\subsection{Cramér-Rao Bounds on the Channel Components}
Define the concatenated vector $\bm{h}_{jt} = [\bm{h}_s^\top, \bm{h}_d^\top]^\top$ and $\bm{h}^\diamond_{jt} = [\bm{h}_{jt}^\top, \bm{h}_{jt}^\herm]^\top$, and let us denote
\begin{equation}
    \bm{V}\triangleq\frac{\partial\bm{h}_{jt}^\diamond}{\partial \bm{\theta}} = \begin{bmatrix}
        \bm{A}'_{\bm{f}} & \bm{A}_{\bm{f}}&\bm{0}_{N \times L}&\bm{0}_{N \times m}&\bm{0}_{N \times L}\\
        \bm{0}_{N \times m} & \bm{0}_{N \times m}&\bm{D}&\bm{0}_{N \times m}&\bm{0}_{N \times L}\\
        \overline{\bm{A}'_{\bm{f}}}&\bm{0}_{N \times m}&\bm{0}_{N \times L}&\overline{\bm{A}_{\bm{f}}}&\bm{0}_{N \times L}
        \\
        \bm{0}_{N \times m}&\bm{0}_{N \times m}&\bm{0}_{N \times L}&\bm{0}_{N \times m}&\overline{\bm{D}}
    \end{bmatrix}.
\end{equation}
\begin{prop}\label{prop:CRB_sparse_diffuse_channel}
    Assume $G$ independent observations according to Equation \eqref{eq:channel_model}, $\vb{\bm{n}}^{(g)} \sim \mathcal{CN}(\bm{0}, \sigma^2\bm{I}_N)$ and $\bm{S}_g^\herm \bm{S}_g = \bm{I}_N$. If $3m+2L \leq 2N$ then any unbiased estimator $\hat{\bm{h}}_{jt}^\diamond$ of $\bm{h}_{jt}^\diamond$ satisfies
    \begin{equation}
        \mathbb{E}\left[\left(\bm{h}_{jt}^\diamond-\hat{\bm{h}}_{jt}^\diamond\right)\left(\bm{h}_{jt}^\diamond-\hat{\bm{h}}_{jt}^\diamond\right)^{\herm}\right] \succeq  \frac{\sigma^2}{G} \bm{V}\left(\bm{U}^{\herm}\bm{U}\right)^{-1}\bm{V}^{\herm}.
        \label{eq:crb_channel_s_d}
    \end{equation}
\end{prop}
It is shown in Section~\ref{sec:numerical} that the bounds of Propositions~\ref{prop:CRB_channel} and~\ref{prop:CRB_sparse_diffuse_channel} closely approximate the performance of the HALS estimator, especially in the high-SNR regime. 
\section{Numerical Results}

\label{sec:numerical}
\subsection{Experimental Setup}
Our experimental setup is as follows. WLOG, assume a single snapshot, $G=1$. Rayleigh fading~\cite{michelusi2012uwb} is assumed for the path gains, so that $\alpha_i^\star \sim \mathcal{CN}\left(0, e^{-\omega\frac{\tau_i}{\Delta T}}\right)$ for $i = \llbracket 1, m \rrbracket$, and $\gamma_r^\star \sim \mathcal{CN}\left(0, \beta e^{-\omega r}\right)$ for $r = \llbracket 0, L-1 \rrbracket$.
The noise is drawn from $\bm{n} \sim \mathcal{CN}\left(\bm{0}, \sigma^2 \bm{I}_N\right)$.  {Our data symbols are drawn uniformly from the QPSK constellation.}
We assume that $\bm{D}$ is known. $E_d$ is chosen to take the value of the expected amplitude of $\bm{\gamma}^\star$.
{For each channel realization,} we define the signal-to-noise ratio $\mbox{SNR(dB)} = 10\log_{10}\frac{\left\Vert\bm{h}^\star\right\Vert_2^2}{N\sigma^2}$. 
For our HALS algorithm, $\hat{\bm{h}}_{\mathrm{HALS}} = \hat{\bm{h}}_{s, \mathrm{db}}+\bm{D}\hat{\bm{\gamma}}$.  Our method is compared with the solution $\hat{\bm{h}}_{ANM}=\hat{\bm{h}}_{s, ANM}$ for vanilla ANM~\cite{bhaskar2013atomic}.  Here, the hyperparameters for HALS and vanilla ANM are fine-tuned to give an estimate with the correct number of sparse components and to yield the best MSE performance as in \cite{tang_near_2015}. Furthermore, we construct a  genie-aided scheme which, in addition to knowledge of the support of the diffuse components, $\bm{D}$, also knows the support of the sparse components, $\bm{A}_{\bm{f}^\star}$, and the expected energy of the diffuse components, $E_d$:
\begin{align}
    \min\limits_{\bm{\alpha}_{ge}, \bm{\gamma}_{ge}} & \left\Vert\widetilde{\bm{y}}-\bm{A}_{\bm{f}^\star}\bm{\alpha}_{ge}-\bm{D}\bm{\gamma}_{ge}\right\Vert_2^2+\mu\left\Vert\bm{\gamma}_{ge}\right\Vert_2^2.\tag{P4}
\end{align}

The genie estimator has a closed-form solution for amplitude estimation. Given $\hat{\bm{h}}_{s, ge} = \bm{A}_{\bm{f}^\star}\hat{\bm{\alpha}}_{ge}$ and $\hat{\bm{h}}_{d, ge} = \bm{D}\hat{\bm{\gamma}}_{ge}$,
\begin{equation}
  \begin{bmatrix}
         \hat{\bm{h}}_{s, ge}  \\
          \hat{\bm{h}}_{d, ge}\end{bmatrix}= \begin{bmatrix}
         \bm{P}_{\bm{A}_{\bm{f}^\star}}\left(\bm{I}-\bm{D}\bm{T}^{-1}\bm{D}^\herm\bm{P}_{\bm{A}_{\bm{f}^\star}^{\perp}}\right)  \\
         \bm{D}\bm{T}^{-1}\bm{D}^\herm\bm{P}_{\bm{A}_{\bm{f}^\star}^{\perp}}\end{bmatrix}\widetilde{\bm{y}},
    \label{eq:genie_est}
\end{equation}
where $\bm{P}_{\bm{A}_{\bm{f}^\star}} = \bm{A}_{\bm{f}^\star}\bm{A}_{\bm{f}^\star}^{\dagger}$, $\bm{P}_{\bm{A}_{\bm{f}^\star}^{\perp}} = \bm{I}-\bm{P}_{\bm{A}_{\bm{f}^\star}}$, $\bm{T}\triangleq \bm{D}^\herm\bm{P}_{\bm{A}_{\bm{f}^\star}^{\perp}}\bm{D}+\mu\bm{I}$. The genie method provides an estimate $\hat{\bm{h}}_{\mathrm{ge}} = \hat{\bm{h}}_{s,\mathrm{ge}}+\hat{\bm{h}}_{d,\mathrm{ge}}$, which serves as a benchmark for performance when the frequencies can be accurately estimated. The least-squares estimator of the channel is given by $\hat{\bm{h}}_{\mathrm{LS}}=\widetilde{\bm{y}}$. The convex problems are solved with \texttt{CVX} \cite{cvx,gb08}. 

The simulation results are {\bf averaged} over at least 50 synthetic channel realizations and 100 noise realizations per channel.
The performance metric on \emph{overall channel} estimation is the normalized mean square error (NMSE): $\mbox{NMSE}(\bm{h}^\star)=\mathbb{E} \left\{\frac{\left\Vert\bm{h}^\star-\hat{\bm{h}}\right\Vert_2^2}{\left\Vert\bm{h}^\star\right\Vert_2^2} \right\}$, {where the expectation is over the noise and the channel realizations.} Similarly, for the estimates on the sparse and diffuse parts of the aHSD channel, we also adapt normalized metrics, $ \mbox{NMSE}(\bm{h}_s^\star)$ and $ \mbox{NMSE}(\bm{h}_d^\star)$. For consistency with the NMSE, the corresponding CRBs are also normalized.

 The metric on \emph{diffuse amplitude} estimates is the mean-squared error (MSE): $\text{MSE}\left(\bm{\gamma}^\star\right) = \frac{1}{L}\mathbb{E}\{\left\Vert\bm{\gamma}^\star-\hat{\bm{\gamma}}\right\Vert_2^2\}$, with expectation as noted previously. For the sparse components, denote the set of frequency estimates as $\hat{\mathcal{T}}$.
 Our experiments revealed that 
 the estimation errors of the frequency components are dominated by outliers arising from very weak amplitudes of resolvable paths. As is often done in non-sparse estimation, we reject outliers~\cite{henninger2022probabilistic, muthineni2024outlier}. Following our model assumption that diffuse components have far lower energy, we define the thresholded frequency estimates subset $\hat{\mathcal{T}}_{th} = \left\{\hat{f}_i\in \hat{\mathcal{T}}: \left\vert \hat{\alpha_i}\right\vert\geq \max\limits_{0 \leq d \leq L-1}\left\{\left\vert\hat{\gamma}_d\right\vert\right\}\right\}$. 
 We match the estimates with the ground truth as
\begin{equation}
\bm{\pi}^\star=\argmin\limits_{\bm{\pi}}\left\{\sum\limits_{i = 1}^{\left\vert\hat{\mathcal{T}}\right\vert}\mathbbm{1}\left(\hat{f}_i \in\hat{\mathcal{T}}_{th}\right)\left\vert \hat{f}_i - f^\star_{\pi_i}\right\vert^2\right\},
\end{equation}
where $\bm{\pi}^\star=\left[\pi_1^{\star}, \cdots, \pi_{\left\vert \hat{\mathcal{T}}\right\vert}^{\star}\right]^\top$ is the optimal permutation of $t \in \llbracket 1, \left\vert \hat{\mathcal{T}}\right\vert \rrbracket$ to match the truth and estimate.

The threshold metrics for the frequencies and sparse amplitudes estimates are defined as follows
\begin{subequations}
    \begin{align}
    \text{MSE}_{th}\left(\bm{\vartheta}^\star\right)& = \left\vert\hat{\mathcal{T}}\right\vert^{-1}\sum\limits_{i = 1}^{\left\vert\hat{\mathcal{T}}\right\vert}\mathbbm{1}\left(\hat{f}_i \in\hat{\mathcal{T}}_{th}\right)\left\vert \hat{\vartheta}_i - \vartheta^\star_{\pi_i^{\star}}\right\vert^2;
     \\
    \text{CRB}_{th}\left(\bm{\vartheta}^\star\right)& = \left\vert\hat{\mathcal{T}}\right\vert^{-1}\sum\limits_{i = 1}^{\left\vert\hat{\mathcal{T}}\right\vert}\mathbbm{1}\left(\hat{f}_i \in\hat{\mathcal{T}}_{th}\right)\text{CRB}\left(\vartheta^\star_{\pi_i^{\star}}\right),
\end{align}
\end{subequations}
where $\bm{\vartheta}^\star = \bm{\alpha}^\star$ or $\bm{f}^\star$.

The threshold metrics are {\bf averaged} over all Monte Carlo simulations, weighted by the corresponding size of the valid frequency estimates, $\left\vert\hat{\mathcal{T}}_{th}\right\vert$. With thresholding, we eliminate frequency parameters that may be significantly shifted or attenuated by noise or diffuse components. {About $96\%$ of frequency estimates are retained when $\beta=0.01$, while this percentage drops to approximately $86\%$ when $\beta = 0.04$.}
\subsection{Impact of the SNR}
\begin{figure}[t!]
  \centerline{\includegraphics[width=8.5cm]{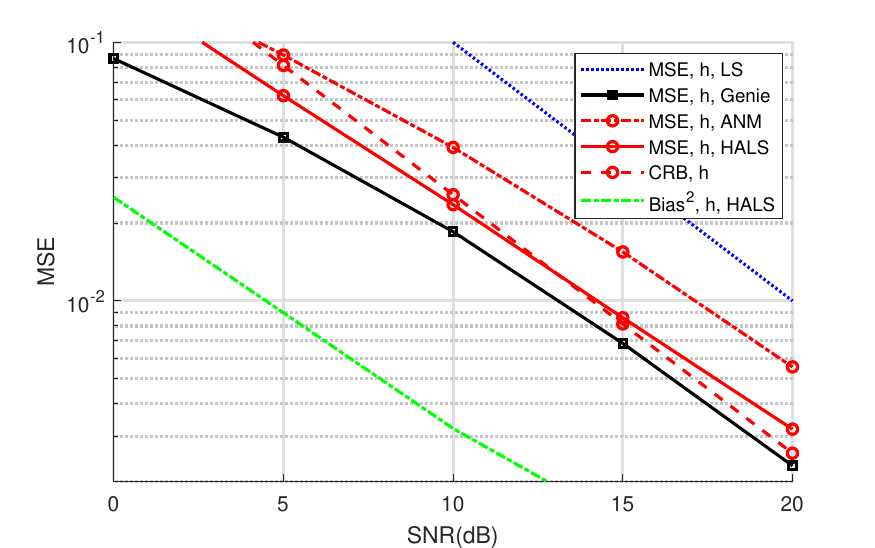}}
  \vspace*{-0.1in}
\caption{Normalized MSE/CRB on the channel coefficients against SNR. Parameter values are: $N=101$, $L=20$, $\omega=0.05$, $\beta = 0.01$, $m =4$ and $\delta = 1/(2L)$.}
\label{fig:exp1}
\vspace{-10pt}
\end{figure}
\begin{figure}[t!]
  \centerline{\includegraphics[width=8.5cm]{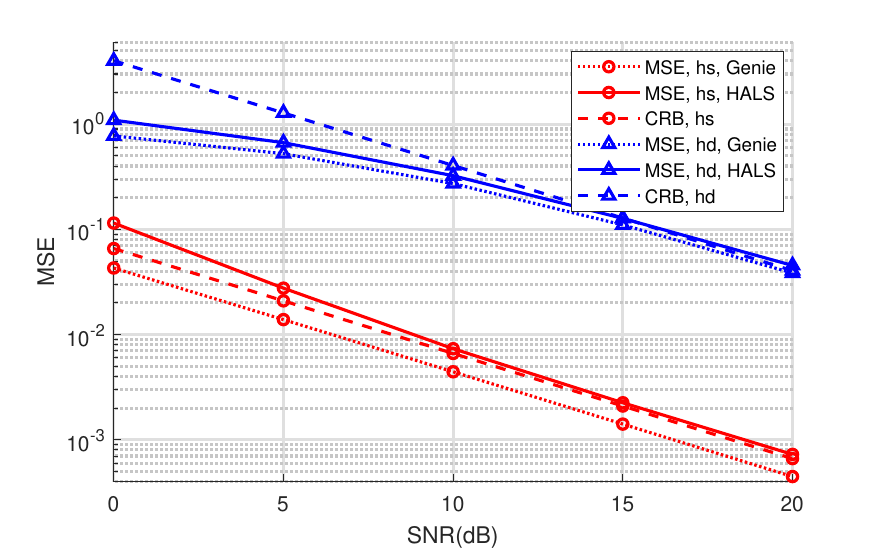}}
      \vspace*{-0.1in}
\caption{Normalized MSE/CRB on the sparse/diffuse channel coefficients against SNR. Parameter values are: $N=101$, $L=20$, $\omega=0.05$, $\beta = 0.01$, $m =4$ and $\delta = 1/(2L)$.}
\label{fig:exp1.1}
\end{figure}

Figures~\ref{fig:exp1} and~\ref{fig:exp1.1} benchmark the performance of the estimators to recover the channel response $\bm{h}^\star$, as well as the sparse components $\bm{h}_s^\star$ and the diffuse components $\bm{h}_d^\star$, against the CRB. The number of unknown parameters $3m+2L$ is 52, and the relative energy parameter $\beta$ is 0.01. Clearly, as shown in Figure~\ref{fig:exp1}, HALS improves channel estimation compared to the ANM estimator
and, unlike ANM, {HALS approaches the CRB for channel response.}
As the HALS estimator in  Equation \eqref{eq:als_primal_sdp} is {\bf biased} due to the regularization on diffuse components, its performance can be locally lower than CRB. Figure~\ref{fig:exp1.1} shows that the performance of HALS for estimating $\bm{h}_s^\star$ closely follows the CRB in Equation \eqref{eq:crb_channel_s_d}, especially when SNR is high. For the diffuse components, the HALS estimator is much better than the CRB in Equation \eqref{eq:crb_channel_s_d} in the low-SNR regime {due to the impact of bias}. The curves (Genie, HALS, and CRB) for $\bm{h}$ and $\bm{h}_d$ behave similarly as the SNR changes, where curves for $\bm{h}_s$ have different trends. These plots indicate that the estimation error of $\bm{h}_d^\star$ plays a major role in the overall channel estimation error, even though the diffuse components have much lower energy. The reason can be inferred from Figure~\ref{fig:exp1.1} which shows that the normalized performance for the sparse components is much better than that of the diffuse components. Thus, our HALS algorithm is better at estimating sparse components and frequency parameters. To improve our algorithm, we need to consider reducing estimation errors in diffuse components.
\begin{figure}[t!]
    \centering
    \begin{subfigure}[Frequencies, $\bm{f}$.]
        {\includegraphics[width=8.5cm]{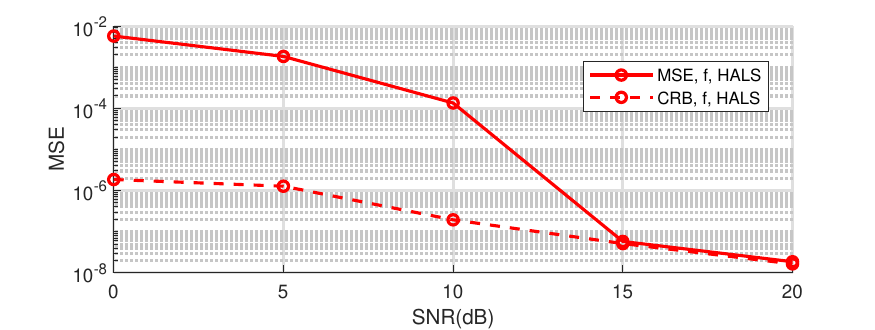}}
        \label{fig:exp2_1}
    \end{subfigure}
    \begin{subfigure}[Sparse amplitudes, $\bm{\alpha}$.]
    {\includegraphics[width=8.5cm]{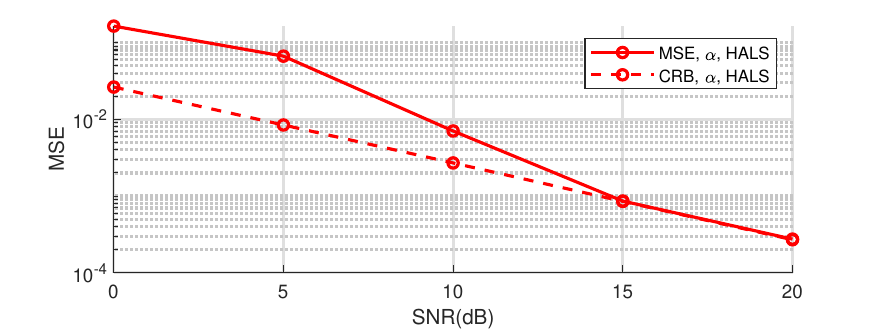}}
        \label{fig:exp2.2}
    \end{subfigure}
    \begin{subfigure}[Diffuse amplitudes, $\bm{\gamma}$.]
        {\includegraphics[width=8.5cm]{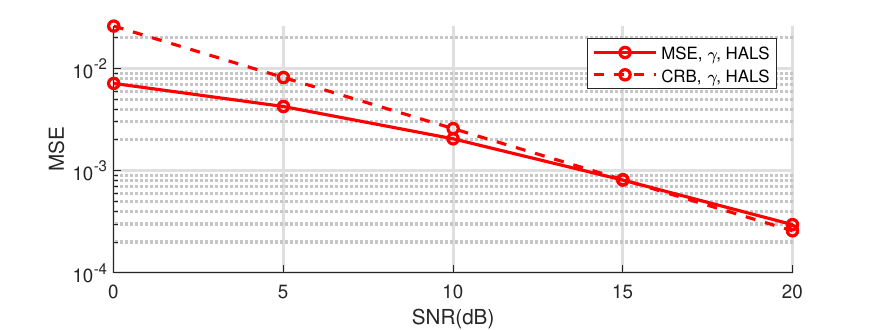}}
        \label{fig:exp2.3}
    \end{subfigure}
      \vspace*{-0.1in}
    \caption{MSE/CRB on channel parameters estimation against SNR. Parameter values are: $N=101$, $L=20$, $\omega=0.05$, $\beta = 0.01$, $m=4$ and $\delta = 1/(2L)$.}
    \vspace{-12pt}
    \label{fig:exp2}
\end{figure}

Figure~\ref{fig:exp2} shows the performance of HALS in estimating the frequencies of $\bm{h}^\star_s$, and amplitudes of $\bm{h}^\star_s$ and $\bm{h}^\star_d$, against CRBs, under different SNR values. The biased estimator HALS closely matches CRBs when estimating frequencies and amplitudes for sparse components in the high-SNR regime.
The errors drop sharply to CRBs at around 15dB. In comparison, the HALS algorithm's error in estimating diffuse amplitudes is better than the CRB at low SNR, but {slightly} exceeds it at higher SNR. This observation is consistent with Figure~\ref{fig:exp1.1} and indicates that the bias in channel estimation arises from $\hat{\bm{\gamma}}$.
\subsection{Impact of the Minimum Separation Between the Delays}
\begin{figure}[t!]
  \centerline{\includegraphics[width=8.5cm]{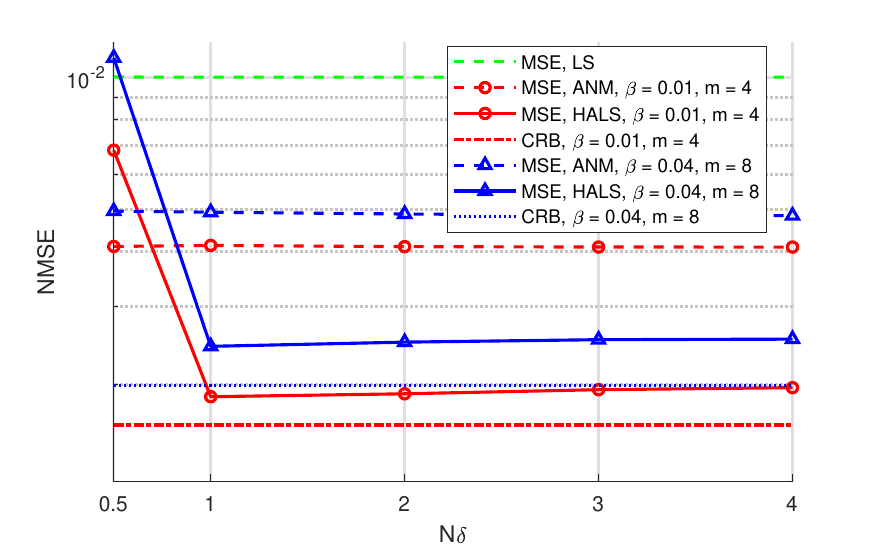}}
      \vspace*{-0.1in}
\caption{Normalized MSE/CRB on the channel coefficients against the separation values, $N\delta$. Parameter values are: $N=161$, $L=20$, $\omega=0.05$, $\beta = 0.01$, $m=4$ and $\text{SNR} = 20\text{dB}$.}
\label{fig:exp3.0}
\end{figure}
In Figure~\ref{fig:exp3.0}, we show how the minimum separation $N\delta$ influences the NMSE of the channel estimate. {When $N\delta \in [1,4]$ all the curves barely change.} As expected, HALS' performance is closer to the CRB than that of ANM. Increasing values of $m$ and $\beta$ increase the error, but HALS remains close to the CRB. However, if $N\delta$ changes to 0.5, the performance of HALS deteriorates dramatically. As illustrated in the CRB analysis (Section~\ref{sec:CRB}), the minimum separation influences the accuracy of the frequency estimate.  The debiasing step in HALS, which aims to eliminate bias, requires an accurate frequency estimate. If the estimate is inaccurate, we are projecting onto the wrong subspace, leading to even larger bias. As a result, HALS is more strongly affected when the minimal separation, $N\delta$, is quite small. ANM does not need debiasing when estimating aHSD channels.
\begin{figure}[t!]
    \centering
    \begin{subfigure}[Frequencies, $\bm{f}$.]
        {\includegraphics[width=8.5cm]{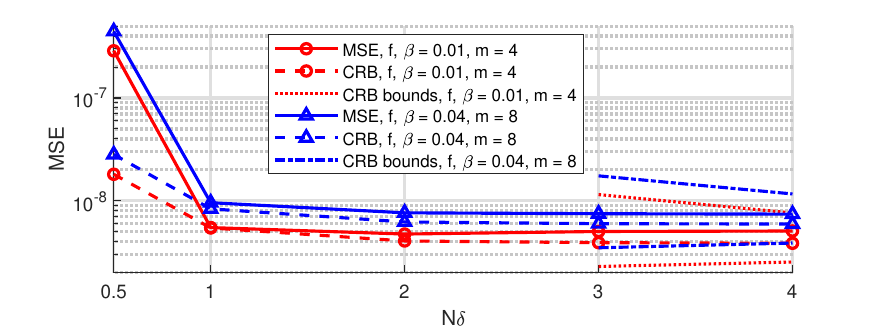}}
        \label{fig:exp3_1}
    \end{subfigure}
    \begin{subfigure}[Sparse amplitudes, $\bm{\alpha}$.]
        {\includegraphics[width=8.5cm]{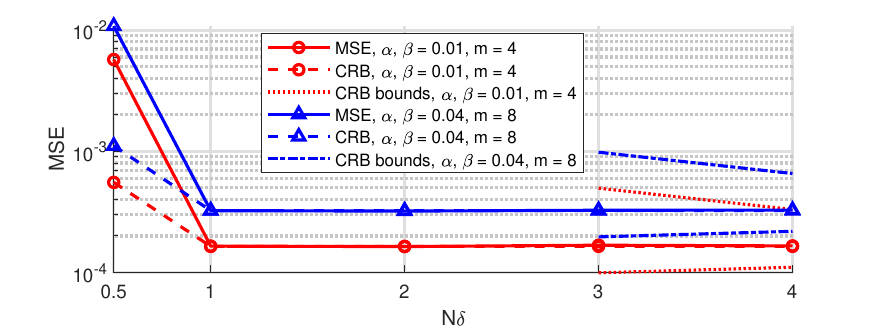}}
        \label{fig:exp3_2}
    \end{subfigure}
    \begin{subfigure}[Diffuse amplitudes, $\bm{\gamma}$.]
        {\includegraphics[width=8.5cm]{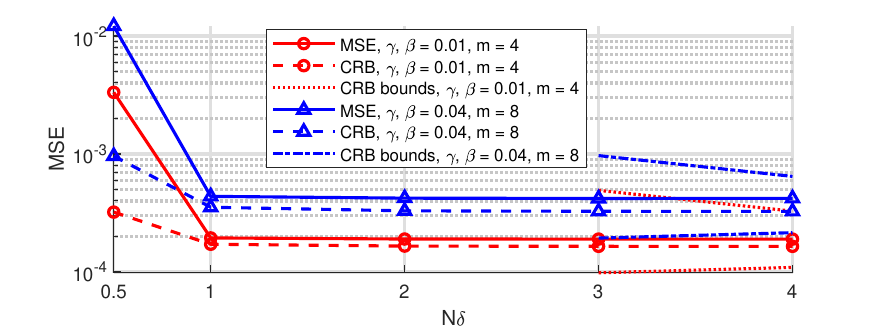}}
        \label{fig:exp3_3}
    \end{subfigure}
      \vspace*{-0.1in}
    \caption{MSE/CRB on channel parameters estimation against the separation values, $N\delta$. Parameter values are: $N=161$, $L=20$, $\omega=0.05$, $\beta = 0.01$, $m=4$ and $\text{SNR} = 20\text{dB}$.}
    \label{fig:exp3}
\end{figure}

Figure~\ref{fig:exp3} shows how HALS performs in estimating $\bm{f}^\star$, $\bm{\alpha}^\star$ and $\bm{\gamma}^\star$ against the CRB, at varying values of $N\delta$. As expected, at $N\delta=0.5$, the frequency estimate is poor. The amplitude estimate is also affected when frequencies cannot be accurately estimated. Similarly to Figure~\ref{fig:exp3.0}, we observe that the curves are flat and HALS approaches the CRBs when $N \delta \in[1,4]$, which certifies the robustness of HALS against decreasing minimal separations within a certain range. 
The bounds on CRBs in Proposition~\ref{prop:CRB_bounds} are also verified here, where the gap between lower and upper bounds narrows as $N\delta$ gets larger. 
Furthermore, we observe that the MSE and CRB curves overlap for the sparse amplitudes estimate. The CRBs better characterize sparse component estimation as the sparse amplitudes estimates are debiased. 

\subsection{Impact of the Number of Resolvable Paths}
\begin{figure}[t!]
    \centering
    \begin{subfigure}[Frequencies, $\bm{f}$.]
        {\includegraphics[width=8.5cm]{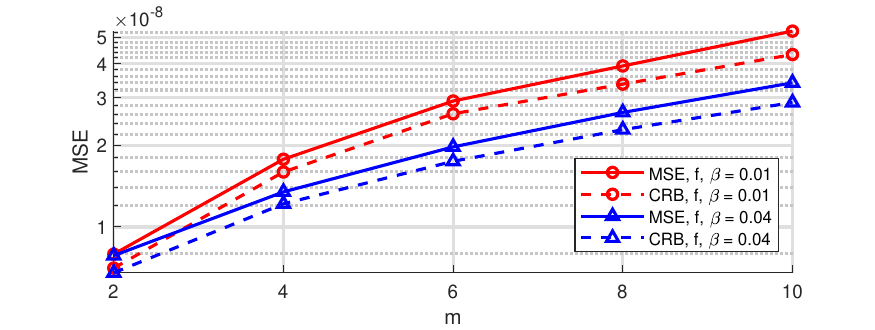}}
        \label{fig:exp4_1}
    \end{subfigure}
    \begin{subfigure}[Sparse amplitudes, $\bm{\alpha}$.]
        {\includegraphics[width=8.5cm]{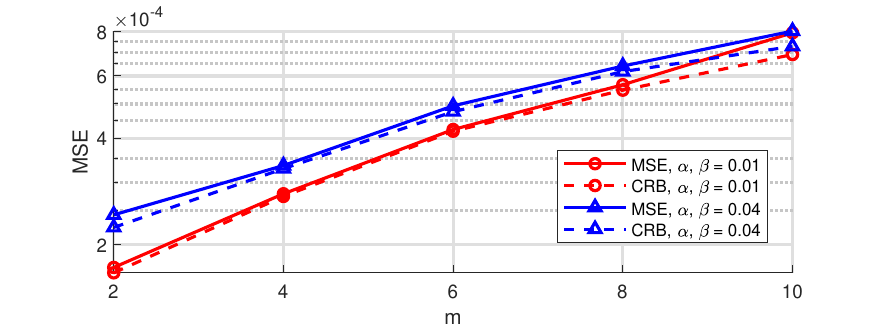}}
        \label{fig:exp4_2}
    \end{subfigure}
    \begin{subfigure}[Diffuse amplitudes, $\bm{\gamma}$.]
        {\includegraphics[width=8.5cm]{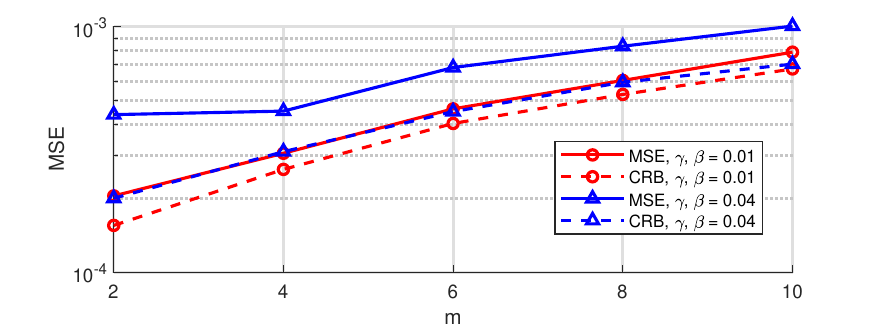}}
        \label{fig:exp4_3}
    \end{subfigure}
      \vspace*{-0.1in}
    \caption{MSE/CRB on channel parameters estimation against the number of sparse components, $m$. Parameter values are: $N=101$, $L=20$, $\omega=0.05$, $\delta = 1/(2L)$ and $\text{SNR} = 20\text{dB}$.}
    \label{fig:exp4}
\end{figure}
According to Figures~\ref{fig:exp3.0} and~\ref{fig:exp3}, changing from $(m, \beta) = (4, 0.01)$ to $(m, \beta) = (8, 0.04)$, the MSE worsens, but all curves still exhibit similar behavior, showing the robustness of HALS with respect to variations in $m$ and $\beta$. In Figure~\ref{fig:exp4}, we further investigate how parameter estimation with HALS varies for different values of $m$ and $\beta$, where the minimum separation is fixed. We see that the MSE and CRB curves increase as $m$ increases, while $\beta$ has little impact on the estimates of the sparse components. The performances of frequency and sparse amplitude estimates are close to the CRBs, indicating that HALS is also robust to variations in $m$ and $\beta$ when estimating sparse components. In comparison, the diffuse amplitude estimates are vulnerable to the increasing magnitude of diffuse components, which is proportional to $\beta$.
\subsection{Experiments on Real Channel Data}
\begin{figure}[t!]
  \centerline{\includegraphics[width=8.5cm]{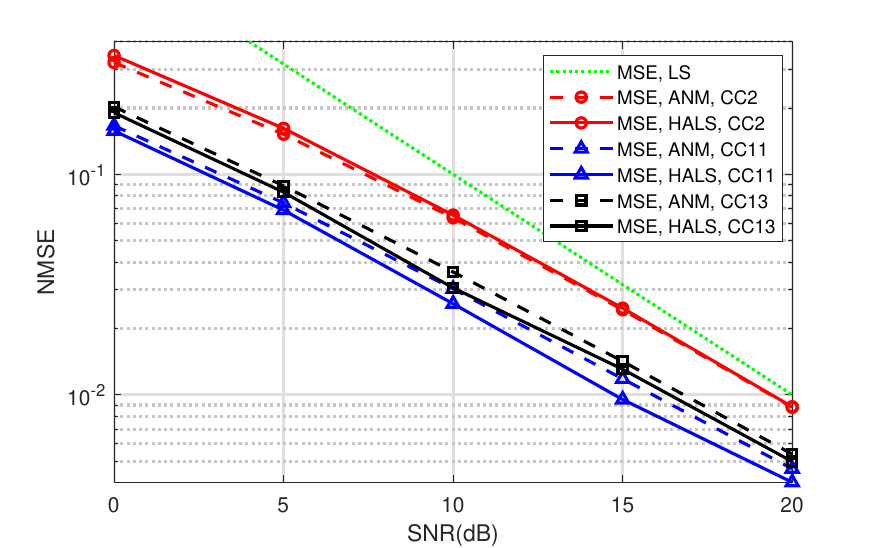}}
    \vspace*{-0.1in}
\caption{MSE/CRB on real channel estimation against the SNR values. Parameter values are: $N=201$, $\hat{L}=201$. CC2, CC11 and CC13 represent the 2nd, 11th and 13th channel in ``Child Care'' dataset from NIST\cite{uwbdata}.
}
\label{fig:exp5}
\end{figure}
For the experiments on real channels, we take the data from the NIST ``Child Care'' dataset \cite{uwbdata}. In the experiments, we set $\hat{L}$ to the estimated sampling size for the diffuse components in the channel. In the optimization problem, we instead use the estimated basis matrix
\begin{equation}
    \hat{\bm{D}}=\left[\bm{a}\left(0\right), \bm{a}\left(-\frac{1}{\hat{L}}\right), \cdots, \bm{a}\left(-\frac{\hat{L}-1}{\hat{L}}\right)\right]\in \mathbb{C}^{N \times \hat{L}}.
\end{equation}

As shown in Figure~\ref{fig:exp5}, we achieve performance improvements with HALS over ANM when channels better match our model assumptions. For a good aHSD channel, we assume the values of $m$ and $\beta$ are small, and $N\delta$ is large enough. The CC2 channel consists of many spikes and, more importantly, those spikes are not sufficiently separated. As expected, HALS does not provide significant improvement over ANM, and the gap between the LS estimator is narrow too. In contrast, when we apply HALS to channels better matched to the aHSD model, e.g., CC11 and CC13, HALS clearly outperforms ANM. Because the real dataset contains diffuse components with arbitrary delays, using $\hat{\bm{D}}\in \mathbb{C}^{N\times \hat{L}}$ to approximate the support of these components introduces a model mismatch.  However, HALS is never measurably worse than ANM.
Notice that when increasing $\hat{L}$ to reduce the mismatch, the minimal separation will decrease, making it harder to estimate the frequency parameters of sparse components. Therefore, setting the values of $\hat{L}$ properly is important for applying HALS to measured hybrid channels.

Simulations on real data validate the necessity of adopting the proposed aHSD model in environments where data are measured. Together with results on synthetic data, they indicate that better modeling and estimation of diffuse components are important for further performance improvement.
\vspace*{-0.05in}
\section{Conclusions}
\label{sec:conclusions}
In this work, we proposed an aHSD channel framework that maps continuous-time signals to the frequency domain using OFDM signaling. 
Based on the aHSD model, we also proposed the Hybrid Atomic-Least-Squares (HALS) algorithm, which regularizes sparse and diffuse terms with the atomic norm and the $\ell_2$ norm, respectively. We  provided a theoretical analysis on the relevant dual problems and exploited the properties of primal and dual solutions. We then derived a CRB for parameter estimation in our aHSD model and computed upper and lower bounds on the CRB, directly exploiting the aHSD structure with clear dependence on the minimum separation between spikes, spike amplitudes, and channel length. The analysis employs new tools to accommodate the aHSD structure. The CRB on the parameters leads to a CRB on the channel response. Our numerical results showed that HALS achieves near-CRB performance and has a strong improvement over ANM, especially in the high-SNR regime for channels following the aHSD model. Interestingly, HALS performance is sensitive to the interim quality of diffuse component estimation.
Our numerical results on real data agree with the observations from synthetic data and validate the  aHSD channel model.

\appendices

\section{Proof of Proposition~\ref{prop:exact_recovery}}\label{sec:proof_exact_recovery}
We denote by $\bm{\xi} = \bm{f}^\star \cup \{0, \cdots, \tfrac{L-1}{L}\}$ the union of the support of the sparse and diffuse components. Under the assumption of Proposition~\ref{prop:exact_recovery}, the set $\bm{\xi}$ is sufficiently separated, thus the assumption of~\cite[Prop. 2.3]{fernandez2016super} will hold. Hence, for any $\bm \gamma \in \mathbb{C}^L$, there exists a unique atomic decomposition of the form of Equation \eqref{eq:sparse_est_decomp} that realizes the atomic norm $\norm{\widetilde{\bm{y}} - \bm D \bm \gamma}_{\mathcal{A}}$. Furthermore, that decomposition is of the form
\(
\bm x = \sum\limits_{i=1}^m \alpha_i^\star \bm{a}\left(f_i^\star\right) + \bm{D}\left(\bm \gamma^\star - \bm \gamma \right).
\)
Therefore for any $\bm \gamma \in \mathbb{C}^L$, one has
\begin{equation}\label{eq:minimizer_atomic}
        \forall \bm \gamma \in \mathbb{C}^L, \quad \norm{\widetilde{\bm{y}} - \bm D \bm \gamma}_{\mathcal{A}} = \norm{\bm \alpha^\star}_{1} + \norm{\bm \gamma^\star - \bm \gamma}_{1}.
\end{equation}
It follows immediately from Equation \eqref{eq:minimizer_atomic} that Equation \eqref{eq:atomic_noiseless} is minimized if and only if $\widehat{\bm{\gamma}} = \bm{\gamma}^\star$, which yields the desired statement.

\section{Proof of Proposition~\ref{prop:primal_property}}\label{appendix:prop:primal_property}
We follow the proof structure of \cite[Lemma 1]{bhaskar2013atomic}, where the key idea is to combine the optimality of the objective function with the convexity of the atomic norm.
    
Denote the objective function in Program \eqref{eq:als_primal} as
\begin{equation}
    f\left(\bm{h}_s, \bm{\gamma}\right) = \frac{1}{2}\left\Vert\widetilde{\bm{y}}-\bm{h}_s-\bm{D}\bm{\gamma}\right\Vert_2^2 + \tau \left\Vert\bm{h}_s\right\Vert_{\mathcal{A}}+\frac{\lambda}{2} \left\Vert\bm{\gamma}\right\Vert_2^2.
    \label{eq:als_primal_obj}
\end{equation}
First, note that the objective function is minimized at the optimal solution $\left(\hat{\bm{h}}_s, \hat{\bm{\gamma}}\right)$. Therefore, for any $\bm{h}_s \in \mathbb{C}^N$, $\bm{\gamma}\in \mathbb{C}^L$ and $\alpha\in \mathbb{R}$, the following inequality is satisfied
\begin{equation}
    f\left(\hat{\bm{h}}_s+\alpha\left(\bm{h}_s-\hat{\bm{h}}_s\right), \hat{\bm{\gamma}}+\alpha\left(\bm{\gamma}-\hat{\bm{\gamma}}\right)\right) \geq f\left(\hat{\bm{h}}_s, \hat{\bm{\gamma}}\right). 
    \label{eq:als_primal_min}
\end{equation}

Substituting Equation \eqref{eq:als_primal_obj} into Equation \eqref{eq:als_primal_min} and rearranging the inequality yields
\begin{align}
\MoveEqLeft[0] \alpha^{-1}\tau\left(\left\Vert\hat{\bm{h}}_s+\alpha\left(\bm{h}_s-\hat{\bm{h}}_s\right)\right\Vert_{\mathcal{A}}-\left\Vert\hat{\bm{h}}_s\right\Vert_{\mathcal{A}}\right)&\nonumber
    \\
    &\geq\text{Re}\left<\widetilde{\bm{y}}-\hat{\bm{h}}_s-\bm{D}\hat{\bm{\gamma}}, \bm{h}_s+\bm{D}\bm{\gamma}-\left(\hat{\bm{h}}_s+\bm{D}\hat{\bm{\gamma}}\right)\right> \nonumber
    \\
    & - \lambda\text{Re}\left<\hat{\bm{\gamma}}, \bm{\gamma}-\hat{\bm{\gamma}}\right>-\frac{\alpha}{2}\left\Vert\bm{h}_s+\bm{D}\bm{\gamma}-\left(\hat{\bm{h}}_s+\bm{D}\hat{\bm{\gamma}}\right)\right\Vert_2^2\nonumber
    \\
    & - \frac{\alpha\lambda}{2}\left\Vert\bm{\gamma}-\hat{\bm{\gamma}}\right\Vert_2^2. \label{eq:als_primal_min_reform}
\end{align}

The atomic norm is a convex function, thus we have $\forall \alpha \in \left(0, 1\right]$, and the following inequality is satisfied
\begin{equation}
    \left\Vert\alpha \bm{h}_s+\left(1-\alpha\right)\hat{\bm{h}}_s\right\Vert_{\mathcal{A}} \leq \alpha \left\Vert\bm{h}_s\right\Vert_{\mathcal{A}}+\left(1-\alpha\right)\left\Vert \hat{\bm{h}}_s\right\Vert_{\mathcal{A}},
\end{equation}
which can be rearranged into
\begin{equation}
    \alpha^{-1}\left(\left\Vert\hat{\bm{h}}_s+\alpha\left(\bm{h}_s-\hat{\bm{h}}_s\right)\right\Vert_{\mathcal{A}}-\left\Vert\hat{\bm{h}}_s\right\Vert_{\mathcal{A}}\right) \leq \left\Vert\bm{h}_s\right\Vert_{\mathcal{A}}-\left\Vert\hat{\bm{h}}_s\right\Vert_{\mathcal{A}} \hspace{-3pt},
    \label{eq:atomic_norm_convex}
\end{equation}
where the LHS is in the subgradient of the atomic norm at $\hat{\bm{h}}_s$, if we let $\alpha\rightarrow0$.

Combining  Equation \eqref{eq:als_primal_min_reform} with Equation \eqref{eq:atomic_norm_convex} leads to
\begin{align}
    &\tau\left(\left\Vert\bm{h}_s\right\Vert_{\mathcal{A}}-\left\Vert\hat{\bm{h}}_s\right\Vert_{\mathcal{A}}\right)\nonumber
    \\
    \overset{(a)}{\geq}&\lim\limits_{\alpha \rightarrow 0}\alpha^{-1}\tau\left(\left\Vert\hat{\bm{h}}_s+\alpha\left(\bm{h}_s-\hat{\bm{h}}_s\right)\right\Vert_{\mathcal{A}}-\left\Vert\hat{\bm{h}}_s\right\Vert_{\mathcal{A}}\right)\nonumber
    \\
    \overset{(b)}{\geq} &\lim\limits_{\alpha \rightarrow 0}\left[\text{Re}\left<\widetilde{\bm{y}}-\hat{\bm{h}}_s-\bm{D}\hat{\bm{\gamma}}, \bm{h}_s+\bm{D}\bm{\gamma}-\left(\hat{\bm{h}}_s+\bm{D}\hat{\bm{\gamma}}\right)\right> \right.\nonumber
    \\
    &- \lambda\text{Re}\left<\hat{\bm{\gamma}}, \bm{\gamma}-\hat{\bm{\gamma}}\right>-\frac{\alpha}{2}\left\Vert\bm{h}_s+\bm{D}\bm{\gamma}-\left(\hat{\bm{h}}_s+\bm{D}\hat{\bm{\gamma}}\right)\right\Vert_2^2 \nonumber
    \\
    &\left.-\frac{\alpha\lambda}{2}\left\Vert\bm{\gamma}-\hat{\bm{\gamma}}\right\Vert_2^2 \right]\nonumber
    \\
    \overset{(c)}{=} &\text{Re}\left<\widetilde{\bm{y}}-\hat{\bm{h}}_s-\bm{D}\hat{\bm{\gamma}}, \bm{h}_s+\bm{D}\bm{\gamma}-\left(\hat{\bm{h}}_s+\bm{D}\hat{\bm{\gamma}}\right)\right>\nonumber
    \\
    &- \lambda\text{Re}\left<\hat{\bm{\gamma}}, \bm{\gamma}-\hat{\bm{\gamma}}\right> \nonumber
    \\
    \overset{(d)}{=} &-\text{Re}\left<\widetilde{\bm{y}}-\hat{\bm{h}}_s-\bm{D}\hat{\bm{\gamma}}, \hat{\bm{h}}_s\right> - \lambda\text{Re}\left<\hat{\bm{\gamma}}, \bm{\gamma}-\hat{\bm{\gamma}}\right>
    \nonumber
    \\
    &+ \text{Re}\left<\widetilde{\bm{y}}-\hat{\bm{h}}_s-\bm{D}\hat{\bm{\gamma}}, \bm{h}_s+\bm{D}\bm{\gamma}-\bm{D}\hat{\bm{\gamma}}\right> \label{eq:als_primal_min_reform2.5},
\end{align}
where $(a)$ holds  $\forall \alpha\in (0, 1]$ because we multiply both the LHS and RHS of  Equation \eqref{eq:atomic_norm_convex} by 
a non-negative number $\tau$. Equality $(b)$ is derived by plugging Equation \eqref{eq:als_primal_min_reform} in the equation above. Equality $(c)$ holds because we eliminate terms that go to zero as $\alpha\rightarrow0$; finally, $(d)$ is due to the linearity property of the inner product.

The  LHS and RHS Inequality \eqref{eq:als_primal_min_reform2.5} can be rearranged to yield,
\begin{align}
    &\tau\left\Vert\hat{\bm{h}}_s\right\Vert_{\mathcal{A}}-\text{Re}\left<\widetilde{\bm{y}}-\hat{\bm{h}}_s-\bm{D}\hat{\bm{\gamma}}, \hat{\bm{h}}_s\right>\nonumber
    \\
  {\leq} & \tau\left\Vert\bm{h}_s\right\Vert_{\mathcal{A}}-\text{Re}\left<\widetilde{\bm{y}}-\hat{\bm{h}}_s-\bm{D}\hat{\bm{\gamma}}, \bm{h}_s+\bm{D}\bm{\gamma}-\bm{D}\hat{\bm{\gamma}}\right>\nonumber
    \\
    &+ \lambda\text{Re}\left<\hat{\bm{\gamma}}, \bm{\gamma}-\hat{\bm{\gamma}}\right> \nonumber
    \\
    \overset{(e)}{=}& \underbrace{\left(\tau\left\Vert\bm{h}_s\right\Vert_{\mathcal{A}}-\text{Re}\left<\widetilde{\bm{y}}-\hat{\bm{h}}_s-\bm{D}\hat{\bm{\gamma}}, \bm{h}_s\right>\right)}_{\left(\text{A}\right)}\nonumber
    \\
    &\underbrace{- \text{Re}\left<\bm{D}^{\herm}\left(\widetilde{\bm{y}}-\hat{\bm{h}}_s-\bm{D}\hat{\bm{\gamma}}\right)-\lambda \hat{\bm{\gamma}}, \bm{\gamma}-\hat{\bm{\gamma}}\right>}_{\left(\text{B}\right)}. \label{eq:als_primal_min_reform3}
\end{align}

Equality $(e)$ is due to the conjugate symmetry and linearity property of the inner product.

We use proof by contradiction to show that the RHS of Equation \eqref{eq:als_primal_min_reform3} cannot be minus infinity. Suppose the RHS of Equation \eqref{eq:als_primal_min_reform3} is minus infinity, which indicates that the LHS is also minus infinity. Then either $\hat{\bm{\gamma}}$ or $\hat{\bm{h}}_s$ must be of infinite magnitude. This leads to the fact that the primal objective function in Equation \eqref{eq:als_primal_obj} is infinity, given both $\lambda$ and $\tau$ are positive. Clearly, we observe that $f\left(\bm{0}, \bm{0}\right)$ is finite and thus smaller than the optimal value of the objective function. This leads to a contradiction, and our proof is completed.

Recall that Equation \eqref{eq:als_primal_min_reform3} holds for any $\bm{h}_s\in\mathbb{C}^N$ and $\bm{\gamma}\in \mathbb{C}^L$. Neither term (A) nor (B) can go to minus infinity. Then, the following conditions should be satisfied
\begin{subequations}
    \begin{align}
    \left\Vert\widetilde{\bm{y}}-\hat{\bm{h}}_s-\bm{D}\hat{\bm{\gamma}}\right\Vert_{\mathcal{A}}^* &\leq \tau;
     \label{eq:lemma2_1_appendix}
    \\
    \bm{D}^{\herm}\left(\widetilde{\bm{y}}-\hat{\bm{h}}_s-\bm{D}\hat{\bm{\gamma}}\right)&=\lambda \hat{\bm{\gamma}},
    \label{eq:lemma2_2_appendix}
\end{align}
\end{subequations}
where Equation \eqref{eq:lemma2_1_appendix} holds, since the dual norm definition indicates that $\exists \bm{u}_s \in \mathbb{C}^N$ such that
\begin{equation}
\text{Re}\left<\widetilde{\bm{y}}-\hat{\bm{h}}_s-\bm{D}\hat{\bm{\gamma}}, \bm{u}_s\right>= \left\Vert\widetilde{\bm{y}}-\hat{\bm{h}}_s-\bm{D}\hat{\bm{\gamma}}\right\Vert_{\mathcal{A}}^* \left\Vert\bm{u}_s\right\Vert_{\mathcal{A}}.
\end{equation}

If we take $\bm{h}_s = k \bm{u}_s$ with $k \in \mathbb{R}$, then, term (A) can be rewritten as $k\left(\tau-\left\Vert\widetilde{\bm{y}}-\hat{\bm{h}}_s-\bm{D}\hat{\bm{\gamma}}\right\Vert_{\mathcal{A}}^*\right)\left\Vert\bm{u}_s\right\Vert_{\mathcal{A}}$. If Equation \eqref{eq:lemma2_1_appendix} is not satisfied, the term (A) will go minus infinity as $k\rightarrow\infty$. Thus Equation \eqref{eq:lemma2_1_appendix} holds.

Equation \eqref{eq:lemma2_2_appendix} holds as term (B) cannott be minus infinity. It is easy to see that we can find some $\bm{\gamma} \in \mathbb{C}^L$ to contradict this conclusion if the first vector in the inner product is not a zero vector. Therefore, Equation \eqref{eq:lemma2_2_appendix} should be satisfied.

When $\bm{h}_s = \bm{0}$ and $\bm{\gamma} = \hat{\bm{\gamma}}$, Equation \eqref{eq:als_primal_min_reform3} indicates that
\begin{equation}
    \tau\left\Vert\hat{\bm{h}}_s\right\Vert_{\mathcal{A}}-\text{Re}\left<\widetilde{\bm{y}}-\hat{\bm{h}}_s-\bm{D}\hat{\bm{\gamma}}, \hat{\bm{h}}_s\right> \leq 0.
    \label{eq:dual_norm_conclusion}
\end{equation}

The dual norm inequality and Equation \eqref{eq:lemma2_1_appendix} show that $\tau\left\Vert\hat{\bm{h}}_s\right\Vert_{\mathcal{A}}-\text{Re}\left<\widetilde{\bm{y}}-\hat{\bm{h}}_s-\bm{D}\hat{\bm{\gamma}}, \hat{\bm{h}}_s\right> \geq 0$. Therefore, with Equation \eqref{eq:dual_norm_conclusion}, Equation \eqref{eq:lemma2_3} should be satisfied.

\section{Proof of Proposition~\ref{prop:dual_prolem}}\label{appendix:prop:dual_prolem}

To formulate the dual problem of $\text{(P1)}$, we follow the proof structure of \cite[Lemma 2]{bhaskar2013atomic}.

First, an equivalent constrained problem of Program \eqref{eq:als_primal} is given by
\begin{align}
    \min\limits_{\bm{h}_s, \bm{\gamma}, \bm{x}, \bm{b} \in \mathbb{C}^N} & \frac{1}{2}\left\Vert\widetilde{\bm{y}}-\bm{h}_s-\bm{D}\bm{\gamma}\right\Vert_2^2 + \tau \left\Vert\bm{x}\right\Vert_{\mathcal{A}}+\frac{\lambda}{2} \left\Vert\bm{b}\right\Vert_2^2\nonumber
    \\
    \text{s.t.  } & \bm{x} = \bm{h}_s \text{ and } \bm{b} = \bm{\gamma}. \label{eq:als_primal_constrained}
\end{align}
By introducing two Lagrangian multipliers $\bm{z}$ and $\bm{v}$, the Lagrangian of this constrained problem is formulated as
\begin{align}
    \MoveEqLeft[0] \mathcal{L}\left(\bm{h}_s, \bm{\gamma}, \bm{x}, \bm{b};\bm{z}, \bm{v}\right)& \nonumber
    \\
    & \overset{(a)}{=} \frac{1}{2}\left\Vert\widetilde{\bm{y}}-\bm{h}_s-\bm{D}\bm{\gamma}\right\Vert_2^2 + \tau \left\Vert\bm{x}\right\Vert_{\mathcal{A}}+\frac{\lambda}{2} \left\Vert\bm{b}\right\Vert_2^2\nonumber
    \\
    &\quad -\text{Re}\left<\bm{z}, \bm{x}-\bm{h}_s\right>-\text{Re}\left<\bm{v}, \bm{b}-\bm{\gamma}\right> \nonumber
    \\
    &\overset{(b)}{=} \underbrace{\frac{1}{2}\left\Vert\widetilde{\bm{y}}\right\Vert_2^2}_{\left(\text{A}\right)}+\underbrace{\left[\frac{1}{2}\left\Vert\bm{h}_s+\bm{D}\bm{\gamma}-\left(\widetilde{\bm{y}}-\bm{z}\right)\right\Vert_2^2\right.}_{\left(\text{B}\right)}\nonumber
    \\
    &\quad -\underbrace{\left.\frac{1}{2}\left\Vert\widetilde{\bm{y}}-\bm{z}\right\Vert_2^2+\text{Re}\left\{\left<\bm{v}-\bm{D}^{\herm}\bm{z}, \bm{\gamma}\right>\right\}\right]}_{\left(\text{B}\right)}\nonumber
    \\
    & \quad + \underbrace{\left[\tau\left\Vert\bm{x}\right\Vert_{\mathcal{A}}-\text{Re}\left<\bm{z}, \bm{x}\right>\right]}_{\left(\text{C}\right)}+\underbrace{\left[\frac{\lambda}{2}\left\Vert\bm{b}-\frac{1}{\lambda}\bm{v}\right\Vert_2^2-\frac{1}{2\lambda}\left\Vert\bm{v}\right\Vert_2^2\right]}_{\left(\text{D}\right)},
    \label{eq:als_lagrangian}
\end{align}
where in $(a)$, the Lagrangian is constructed via the dual variables associated with the equality constraint. In $(b)$, the inner products and norms are expanded and rearranged into four parts (A)-(D), and the square terms in (B) and (D) are completed.

Then the infimum of the Lagrangian over $\bm{h}_s$, $\bm{\gamma}$, $\bm{x}$ and $\bm{b}$ is derived as 
\begin{align}
    \MoveEqLeft[0] \inf\limits_{\bm{h}_s, \bm{\gamma}, \bm{x}, \bm{b}}\mathcal{L}\left(\bm{h}_s, \bm{\gamma}, \bm{x}, \bm{b};\bm{z}, \bm{v}\right) &\nonumber
    \\
    & \overset{(c)}{=} \underbrace{\frac{1}{2}\left\Vert\widetilde{\bm{y}}\right\Vert_2^2}_{\left(\text{A}\right)}+\inf\limits_{\bm{h}_s, \bm{\gamma}}\underbrace{\left[\frac{1}{2}\left\Vert\bm{h}_s+\bm{D}\bm{\gamma}-\left(\widetilde{\bm{y}}-\bm{z}\right)\right\Vert_2^2\right.}_{\left(\text{B}\right)}\nonumber
    \\
    &\quad -\underbrace{\left.\frac{1}{2}\left\Vert\widetilde{\bm{y}}-\bm{z}\right\Vert_2^2+\text{Re}\left<\bm{v}-\bm{D}^{\herm}\bm{z}, \bm{\gamma}\right>\right]}_{\left(\text{B}\right)}+\nonumber
    \\
    &\quad\inf\limits_{\bm{x}}\underbrace{\left[\tau\left\Vert\bm{x}\right\Vert_{\mathcal{A}}-\text{Re}\left<\bm{z}, \bm{x}\right>\right]}_{\left(\text{C}\right)}+\inf\limits_{\bm{b}}\underbrace{\left[\frac{\lambda}{2}\left\Vert\bm{b}-\frac{1}{\lambda}\bm{v}\right\Vert_2^2-\frac{\left\Vert\bm{v}\right\Vert_2^2}{2\lambda}\right]}_{\left(\text{D}\right)},
    \label{eq:als_lag_inf}
\end{align}
where $(c)$ holds because term (A) is a constant and term (B), (C) and (D) only contains variables $\left(\bm{h}_s, \bm{\gamma}\right)$, $\bm{x}$ and $\bm{b}$ respectively. 

We see that the infimum of the Lagrangian is separable into four parts. Term (A) is constant, whose infimum is itself. Term (C) is a function of $\bm{x}$ only. According to \cite{bhaskar2013atomic}, the infimum of (C) over $\bm{x}$ leads to an indicator function, which is equivalent to the constraint $\left\Vert\bm{z}\right\Vert_\mathcal{A}^*\leq \tau$. We move the terms containing $\bm{b}$ only to (D), the infimum is achieved at $\bm{b} = \frac{1}{\lambda}\bm{v}$. All terms containing $\left(\bm{h}_s, \bm{\gamma}\right)$ are arranged in part (B). It can be shown that $\bm{v}-\bm{D}^{\herm}\bm{z} = \bm{0}$, or the infimum will go minus infinity. The infimum is then achieved at $\bm{h}_s+\bm{D}\bm{\gamma} = \widetilde{\bm{y}}-\bm{z}$. Thus, the dual problem is formulated as $\text{(P3)}$.

In Equation \eqref{eq:als_primal_constrained}, the primal objective function is convex, because the $\ell_2$ norm and atomic norm are convex, and the linear combination of convex functions multiplied by some positive constant is also convex. The equality constraint is affine and obviously feasible. Then strong duality holds, since Slater’s condition \cite{boyd2004convex} is satisfied. 
\section{Proof of Corollary~\ref{cor:dual_primal_residue}}\label{appendix:cor:dual_primal_residue}
 Denote the dual objective function in Program \eqref{eq:als_dual} as
    \begin{equation}
        g\left(\bm{z}\right)=\frac{1}{2}\left\Vert\widetilde{\bm{y}}\right\Vert_2^2 -\frac{1}{2}\left\Vert\widetilde{\bm{y}}-\bm{z}\right\Vert_2^2-\frac{1}{2\lambda}\left\Vert\bm{D}^{\herm}\bm{z}\right\Vert_2^2,
        \label{eq:dual_obj}
    \end{equation}
    and $\hat{\bm{z}} \triangleq \widetilde{\bm{y}}-\hat{\bm{h}}_s-\bm{D}\hat{\bm{\gamma}}$. We have that
    \begin{align}
    g\left(\hat{\bm{z}}\right)=&\frac{1}{2}\left\Vert\widetilde{\bm{y}}\right\Vert_2^2 -\frac{1}{2}\left\Vert\widetilde{\bm{y}}-\hat{\bm{z}}\right\Vert_2^2-\frac{1}{2\lambda}\left\Vert\bm{D}^{\herm}\hat{\bm{z}}\right\Vert_2^2 \nonumber
        \\
        \overset{(a)}{=}&\frac{1}{2}\left\Vert\hat{\bm{z}} +\hat{\bm{h}}_s+\bm{D}\hat{\bm{\gamma}}\right\Vert_2^2 -\frac{1}{2}\left\Vert\hat{\bm{h}}_s+\bm{D}\hat{\bm{\gamma}}\right\Vert_2^2\nonumber
        \\
        &-\frac{1}{2\lambda}\left\Vert\bm{D}^{\herm}\left(\widetilde{\bm{y}}-\hat{\bm{h}}_s-\bm{D}\hat{\bm{\gamma}}\right)\right\Vert_2^2 \nonumber
        \\
        \overset{(b)}{=}&\frac{1}{2}\left\Vert\hat{\bm{z}} \right\Vert_2^2+ \text{Re}\left<\hat{\bm{z}}, \hat{\bm{h}}_s+\bm{D}\hat{\bm{\gamma}}\right>-\frac{1}{2\lambda}\left\Vert\lambda\hat{\bm{\gamma}}\right\Vert_2^2 \nonumber
        \\
        \overset{(c)}{=}&\frac{1}{2}\left\Vert\hat{\bm{z}} \right\Vert_2^2+ \text{Re}\left<\widetilde{\bm{y}}-\hat{\bm{h}}_s-\bm{D}\hat{\bm{\gamma}}, \hat{\bm{h}}_s\right>\nonumber
        \\
        &+\text{Re}\left<\widetilde{\bm{y}}-\hat{\bm{h}}_s-\bm{D}\hat{\bm{\gamma}}, \bm{D}\hat{\bm{\gamma}}\right>-\frac{\lambda}{2}\left\Vert\hat{\bm{\gamma}}\right\Vert_2^2 \nonumber
        \\
         \overset{(d)}{=}&\frac{1}{2}\left\Vert\hat{\bm{z}} \right\Vert_2^2+\text{Re}\left<\bm{D}^{\herm}\left(\widetilde{\bm{y}}-\hat{\bm{h}}_s-\bm{D}\hat{\bm{\gamma}}\right), \hat{\bm{\gamma}}\right>\nonumber
         \\
         &+ \tau \left\Vert\hat{\bm{h}}_s\right\Vert_{\mathcal{A}}-\frac{\lambda}{2}\left\Vert\hat{\bm{\gamma}}\right\Vert_2^2 \nonumber
         \\
         \overset{(e)}{=}&\frac{1}{2}\left\Vert\widetilde{\bm{y}}-\hat{\bm{h}}_s-\bm{D}\hat{\bm{\gamma}} \right\Vert_2^2+ \tau \left\Vert\hat{\bm{h}}_s\right\Vert_{\mathcal{A}}+\frac{\lambda}{2}\left\Vert\hat{\bm{\gamma}}\right\Vert_2^2
    \end{align}
    where in $(a)$, we substitute $\widetilde{\bm{y}} = \hat{\bm{z}}+\hat{\bm{h}}_s-\bm{D}\hat{\bm{\gamma}}$ in the first two terms and $\hat{\bm{z}} = \widetilde{\bm{y}}-\hat{\bm{h}}_s-\bm{D}\hat{\bm{\gamma}}$ in the last term. Equality $(b)$ is derived by expanding the first $\ell_2$ norm and applying  Equation \eqref{eq:lemma2_2} to the last $\ell_2$ norm. Equality $(c)$ holds due to the linearity property of inner product and $\hat{\bm{z}} = \widetilde{\bm{y}}-\hat{\bm{h}}_s-\bm{D}\hat{\bm{\gamma}}$. Equality $(d)$ holds due to Equation \eqref{eq:lemma2_3} and the conjugate symmetry property of the inner product.
   Equality $(e)$ holds due to Equation  \eqref{eq:lemma2_2}, $\hat{\bm{z}} = \widetilde{\bm{y}}-\hat{\bm{h}}_s-\bm{D}\hat{\bm{\gamma}}$ and $\text{Re}\left<\lambda\hat{\bm{\gamma}}, \hat{\bm{\gamma}}\right> = \lambda\left\Vert\hat{\bm{\gamma}}\right\Vert_2^2$.
    Hence, the dual objective $g\left(\hat{\bm{z}}\right)$ is equal to the primal objective function $f\left(\hat{\bm{h}}_s, \hat{\bm{\gamma}}\right)$, which indicates that $\hat{\bm{z}}$ is a dual optimal and $\left(\hat{\bm{h}}_s, \hat{\bm{\gamma}}\right)$ is a primal optimal.
\section{Proof of Proposition~\ref{prop:CRB_bounds}}\label{appendix:prop:CRB_bounds}
We begin with the proof of the following lemma.
\begin{lemma}\label{lemma:singular_bound}
Let $2N\geq 3(m+L)$ and
\begin{subequations}
    \begin{align}
\bm Q_0 &= \begin{bmatrix}
        \sqrt{2}\bm B_0 & \bm B_1
    \end{bmatrix} \in \mathbb{C}^{N \times 2(m+L)};
    \\
    \bm Q_1 &= \begin{bmatrix}
        \bm B_0 & \bm B_1 & \bm 0 \\
        \overline{\bm B_0} & \bm 0 & \overline{\bm B_1}
    \end{bmatrix} \in \mathbb{C}^{2N \times 3(m+L)},
    \label{eq:Q1}
\end{align}
\end{subequations}
where $\bm B_0\in \mathbb{C}^{N \times (m+L)}$ and $\bm B_1\in \mathbb{C}^{N \times (m+L)}$ are arbitrary complex matrices such that $\bm{Q}_0$ is of full column rank. Then
\begin{subequations}
    \begin{align}
    \sigma_{\min}(\bm{Q}_1)&\geq\sigma_{\min}(\bm{Q}_0); \label{eq:Q_min_thm}
    \\
    \sigma_{\max}(\bm{Q}_1)&\leq\sigma_{\max}(\bm{Q}_0). \label{eq:Q_max_thm}
\end{align}
\end{subequations}

\end{lemma}
\begin{proof}
    To prove Equation \eqref{eq:Q_min_thm}, we start by defining the set $\mathcal{Q}_{\min}$ of minimal singular vectors of the matrix $\bm{Q}_1$ as
    \begin{equation}
        \mathcal{Q}_{\min}=\left\{\bm{q} \in \argmin\limits_{\left\Vert \bm{q}\right\Vert_2^2=1} \left\Vert\bm{Q}_1\bm{q}\right\Vert_2^2\right\}.\label{eq:set_min_sing}
    \end{equation}
    Suppose we have a vector $\bm{a}=\left[\bm{a}_1^\top, \bm{a}_2^\top, \bm{a}_3^\top\right]^\top\in \mathcal{Q}_{\min}$, where $\bm{a}_1, \bm{a}_2, \bm{a}_3\in \mathbb{C}^{m+L}$, then
    \begin{align}
        \left\Vert\bm{Q}_1\bm{a}\right\Vert_2^2&\overset{(a)}{=}\left\Vert\begin{bmatrix}
        \bm B_0\bm{a}_1 + \bm B_1\bm{a}_2 \\
        \overline{\bm B_0}\bm{a}_1+\overline{\bm B_1}\bm{a}_3
    \end{bmatrix}\right\Vert_2^2  \overset{(b)}{=}\left\Vert\begin{bmatrix}
        \bm B_0\overline{\bm{a}_1} + \bm B_1\overline{\bm{a}_3}\\
        \overline{\bm B_0}\overline{\bm{a}_1}+\overline{\bm B_1}\overline{\bm{a}_2}
    \end{bmatrix}\right\Vert_2^2
    \nonumber \\
    &\overset{(c)}{=}\left\Vert\begin{bmatrix}
        \bm B_0 & \bm B_1 & \bm 0 \\
        \overline{\bm B_0} & \bm 0 & \overline{\bm B_1}
    \end{bmatrix}\begin{bmatrix}\overline{\bm{a}_1}^\top & \overline{\bm{a}_3}^\top & \overline{\bm{a}_2}^\top\end{bmatrix}^\top\right\Vert_2^2,
    \end{align}
where in (a) we plug in Equation \eqref{eq:Q1}. (b) holds because $ \bm B_0\overline{\bm{a}_1} + \bm B_1\overline{\bm{a}_3}$ is the conjugate of $\overline{\bm B_0}\bm{a}_1+\overline{\bm B_1}\bm{a}_3$ and $\overline{\bm B_0}\overline{\bm{a}_1}+\overline{\bm B_1}\overline{\bm{a}_2}$ is the conjugate of $ \bm B_0\bm{a}_1 + \bm B_1\bm{a}_2$. (c) holds after the matrix multiplication.
    Thus, $\widetilde{\bm{a}}\triangleq\left[\overline{\bm{a}_1}^\top , \overline{\bm{a}_3}^\top , \overline{\bm{a}_2}^\top\right]^\top\in \mathcal{Q}_{\min}$. Hence, if $\bm{a} \in \mathcal{Q}_{\min}$ with singular value $\sigma^2_{\min}(\bm{Q}_1)$, then
    \begin{align}\label{eq:singular_vector_symetry}
        \bm{Q}_1^\herm\bm{Q}_1\left(\bm{a}+\widetilde{\bm{a}}\right)&=\bm{Q}_1^\herm\bm{Q}_1\bm{a}+\bm{Q}_1^\herm\bm{Q}_1\widetilde{\bm{a}} \nonumber
        \\
        &=\sigma^2_{\min}(\bm{Q}_1)\bm{a}+\sigma^2_{\min}(\bm{Q}_1)\widetilde{\bm{a}},
    \end{align}
    which indicates that $\bm{a}+\widetilde{\bm{a}}$ is also a singular vector of $\bm{Q}_1$ with singular value $\sigma_{\min}(\bm{Q}_1)$.
    Therefore, since $\mathcal{Q}_{\min} \neq \emptyset$, there exists $\bm{q}_{\min}\in \mathcal{Q}_{\min}$ of the form
    \(
        \bm{q}_{\min} = \begin{bmatrix}
        \bm{q}_1^\top&\bm{q}_2^\top&\overline{\bm{q}_2}^\top
        \end{bmatrix}^\top
    \)
    for some $\bm{q}_1\in \mathbb{R}^{m+L}$ and $\bm{q}_2\in \mathbb{C}^{m+L}$. This yields
    \begin{align}
         \sigma^2_{\min}(\bm{Q}_1)&\overset{(d)}{=}\left\Vert \bm{Q}_1\bm{q}_{\min}\right\Vert_2^2 \nonumber
         \overset{(e)}{=}\left\Vert\begin{bmatrix}
        \bm B_0\bm{q}_1 + \bm B_1\bm{q}_2\\
        \overline{\bm B_0}\bm{q}_1+\overline{\bm B_1}\overline{\bm{q}_2}
    \end{bmatrix}\right\Vert_2^2 \nonumber
    \\
    &\overset{(f)}{=}\left\Vert\bm{B}_0\bm{q}_1+\bm{B}_1\bm{q}_2\right\Vert_2^2+\left\Vert\overline{\bm B_0}\bm{q}_1+\overline{\bm B_1}\overline{\bm{q}_2}\right\Vert_2^2 \nonumber
    \\
    &\overset{(g)}{=}2\left\Vert\bm{B}_0\bm{q}_1+\bm{B}_1\bm{q}_2\right\Vert_2^2 \nonumber
    \\
    &\overset{(h)}{=}\left\Vert\begin{bmatrix}
        \sqrt{2}\bm B_0 & \bm B_1 \end{bmatrix}\begin{bmatrix}\bm{q}_1^\top & \sqrt{2}\bm{q}_2^\top \end{bmatrix}^\top\right\Vert_2^2
         \nonumber \\
        &\overset{(i)}{\geq} \sigma^2_{\min}(\bm{Q}_0),
    \end{align}
    where (d) holds since $\bm{q}_{\min}\in \mathcal{Q}_{\min}$ defined by Equation \eqref{eq:set_min_sing}; (e) derives from Equation \eqref{eq:Q1}; Equality (f) holds by the additivity property of the $\ell_2$ norm; Equality (g) holds since $\bm{q}_1=\overline{\bm{q}_1} \in \mathbb{R}^{m+L}$ and $\bm{B}_0\bm{q}_1+\bm{B}_1\bm{q}_2$ is the conjugate of $\overline{\bm B_0}\overline{\bm{q}_1}+\overline{\bm B_1}\overline{\bm{q}_2}$; Equality (h) holds after matrix multiplication; and Inequality (i) holds because $\left\Vert \left[\bm{q}_1^\top , \sqrt{2}\bm{q}_2^\top\right] \right\Vert^2_2 = \left\Vert \bm{q}_1 \right\Vert^2_2+2\left\Vert \bm{q}_2 \right\Vert^2_2 = 1$, so that
    \begin{equation}
        \sigma^2_{\min}(\bm{Q}_0) = \min\limits_{\left\Vert \bm{p}\right\Vert_2^2=1} \left\Vert\bm{Q}_0\bm{p}\right\Vert_2^2 \leq \left\Vert\bm{Q}_0\begin{bmatrix}\bm{q}_1^\top & \sqrt{2}\bm{q}_2^\top \end{bmatrix}^\top\right\Vert_2^2.
    \end{equation}
    Inequality~\eqref{eq:Q_max_thm} can be proven via analogous reasoning.
\end{proof}

Under the hypothesis of Proposition~\ref{prop:CRB_bounds} the FIM is given by $\bm J^{\bm{\theta}} = G\sigma^{-2} \bm{U}^\herm \bm{U}$. To proceed with the proof, we let the scaling matrix $\bm{T}$ be defined by
\begin{align}
    \bm{T} = \begin{bmatrix}
        \sqrt{2} C_N \diag(\bm{\alpha}) & \bm 0 & \bm 0 \\
        \bm 0 &\sqrt{N} \bm{I}_{m+L} & \bm 0 \\
        \bm 0 & \bm 0 & \sqrt{N} \bm{I}_{m+L}
    \end{bmatrix}, 
\end{align}
where $C_N = \norm{\bm{\Lambda} \bm{a}(f) }_2 = \sqrt{\tfrac{\pi^2 N (N-1)(N+1)}{3}}$ for all $f \in [0,1)$ is the norm of the derivative of the complex exponential atoms. Furthermore, introduce $\bm Q = \bm U \bm T^{-1}$. With this notation in place, the CRB matrix $\bm{\Theta}_{\mathrm{CRB}}$ can be written as $\bm{\Theta}_{\mathrm{CRB}} = \left(\bm{J}^{\bm{\theta}}\right)^{-1} =\left(\sigma^2/G\right) \bm{T}^{-1} \left(\bm{Q}^\herm \bm{Q} \right)^{-1}  \left(\bm{T}^{\herm}\right)^{-1}$. The CRB on a specific frequency $\mathrm{CRB}(f_i)$ or sparse channel amplitudes $\mathrm{CRB}(\alpha_i)$ and diffuse channel amplitudes $\mathrm{CRB}(\gamma_r)$ are given by the respective diagonal entries of $\bm{\Theta}_{\mathrm{CRB}}$.
\begin{align}\label{eq:bound_diag_entries}
    \left({\bm{\Theta}_{\mathrm{CRB}}}\right)_{j,j} = \sigma^2 \left(\bm{T}_{j,j}\right)^{-2} \left[\left( \bm Q^\herm \bm Q \right)^{-1}\right]_{ j,j}.
\end{align}

Denote by $\lambda_{\max}(\bm M)$ and $\lambda_{\min}(\bm M)$ the largest and smallest eigenvalues of a Hermitian positive matrix $\bm M$, respectively. Since $\lambda_{\min}(\bm M) \leq M_{i,i} \leq \lambda_{\min}(\bm M)$ for any diagonal entries $M_{j,j}$ of the matrix $\bm M$, Equation \eqref{eq:bound_diag_entries} induces,
\begin{multline}\label{eq:CRB_bound_expand}
         \frac{\sigma^2} {\left(\bm{T}_{j,j}\right)^{2}}  \lambda_{\min}\left(\left(\bm Q^\herm \bm{Q}\right)^{-1}\right) \leq {\bm{\Theta}_{\mathrm{CRB}}}\vert_{j,j} \\
         \leq \frac{\sigma^2} {\left(\bm{T}_{j,j}\right)^{2}} \lambda_{\max}\left(\left(\bm Q^\herm \bm{Q}\right)^{-1}\right),
\end{multline}
for all $j \in \llbracket 1, 3m +2L \rrbracket$, or equivalently
\begin{align}
    \sigma^2  \lambda_{\max}^{-1}(\bm Q^\herm \bm{Q})\leq \left(\bm{T}_{j,j}\right)^{2}\left({\bm{\Theta}_{\mathrm{CRB}}}\right)_{j,j} \leq \sigma^2  \lambda_{\min}^{-1}(\bm Q^\herm \bm{Q}).
\end{align}
To complete the proof, it remains to establish
\begin{subequations}\label{eq:eigen_Q}
    \begin{align}
    \lambda_{\min}(\bm Q^\herm \bm Q) &\geq K_{\max}^{-1} \\
    \lambda_{\max}(\bm Q^\herm \bm Q) &\leq K_{\min}^{-1},
    \end{align}
\end{subequations}
and one can obtain the desired bounds in Equation \eqref{eq:CRB_bounds_theo} by substituting Equation \eqref{eq:eigen_Q} into Equation \eqref{eq:CRB_bound_expand}.

To show that Equation \eqref{eq:eigen_Q} holds, we harness the confluent Vandermonde structure of the matrix $\bm Q$~\cite{gautschi1962inverses} and provide an upper and lower bound of its spectrum as a function of the separation criterion $N\delta$ through a Beurling--Selberg type approximation of the rectangle function by bandlimited functions presented in~\cite{vaaler1985some}. For convenience, denote
$\bm B = {\sqrt{N}}^{-1} [\bm A_{\bm{f}}, \bm D]$, $\bm B^\prime =  C_N^{-1} [\bm \Lambda \bm A_{\bm{f}}, \bm \Lambda \bm D]$, and define
\begin{align}
    \widetilde{\bm Q} = \begin{bmatrix}
        \tfrac{1}{\sqrt{2}}\bm B^\prime & \bm B & \bm 0 \\
        \tfrac{1}{\sqrt{2}}\overline{\bm B^\prime} & \bm 0 & \overline{\bm B}
    \end{bmatrix} \in \mathbb{C}^{2N \times 3(m+L)}.
\end{align}

Since the columns of $\bm{Q}$ are included in those of $\widetilde{\bm Q}$, then
\begin{subequations}\label{eq:bounds_Q}
 \begin{align}
    \lambda_{\min}(\bm{Q}^\herm\bm{Q}) &\geq \lambda_{\min}(\widetilde{\bm{Q}}^\herm\widetilde{\bm{Q}}) = \sigma_{\min}^2(\bm{\widetilde{Q}}) \\
    \lambda_{\max}(\bm{Q}^\herm\bm{Q}) &\leq \lambda_{\max}(\widetilde{\bm{Q}}^\herm\widetilde{\bm{Q}})= \sigma_{\max}^2(\bm{\widetilde{Q}}).
\end{align}   
\end{subequations}

Equations~\eqref{eq:eigen_Q} and~\eqref{eq:bounds_Q} suggests that it is sufficient to show $\sigma_{\max}^2({\bm{\widetilde{Q}}}) \leq K_{\min}^{-1}$ and  $\sigma_{\min}^2({\bm{\widetilde{Q}}}) \geq K_{\max}^{-1}$, which can be done by invoking Lemma~\ref{lemma:singular_bound} and~\cite[Theorem 5]{ferreira2023conditionNumber}.

\renewcommand*{\bibfont}{\normalfont\footnotesize}
\printbibliography
\end{document}